\documentclass{article}

\usepackage{bbm}
\usepackage{amsmath}
\usepackage{amssymb}
\usepackage{xargs}
\usepackage{amsthm}
\usepackage{graphicx}
\usepackage{pstricks,pst-plot}

\newcommand{\estimp}{\theta}
\def\esp{\mathbb E}
\def\Rset{\mathbb{R}}
\def\pr{\mathbb P}
\def\convvague{\stackrel{\mbox{\tiny\rm v}}{\to}} 
\def\rmd{\mathrm{d}} 
\def\constant{\aleph}
\newcommand\1[1]{\mathbbm{1}_{#1}}
\newcommand\uncompact[1]{\overline{\Rset}^{#1}\setminus\{\vectorbold0\}}
\newcommand{\vectorbold}[1]{\boldsymbol{#1}}
\newcommand{\tail}[1]{\bar{#1}}
\newcommandx{\numult}[2][2=]{\boldsymbol{\nu}_{\vectorbold{#1}_{#2}}} 
\newcommandx{\scalingseq}[1][1=n]{c_{#1}}  
\newcommandx\process[1][1=X]{\mathbb{#1}}
\def\convweak{\Rightarrow}   
\newcommand\spaceD{\mathbb{D}}  
\def\cov{\mathrm{cov}}
\newcommand\TEDlimit{T}
\newcommandx{\tepseq}[1][1=n]{u_{#1}}  
\def\convprob{\stackrel{\mbox{\small\tiny p}}{\to}}
\def\var{\mathrm{var}}
\newcommand\statinterseq{k}
\newcommandx\orderstat[2][1=X]{#1_{#2}}
\def\convdistr{\stackrel{\mbox{\tiny\rm d}}{\to}} 

\newtheorem{theorem}{Theorem}
\newtheorem{lemma}[theorem]{Lemma}

\newtheorem{corollary}[theorem]{Corollary}
\newtheorem{proposition}[theorem]{Proposition}
\newtheorem{definition}{Definition}

\newtheorem{remark}{Remark}

\begin{document}
\title{Multivariate Tail Estimation: Conditioning on an extreme event}

\author{Rafa{\l} Kulik\thanks{Corresponding author: rkulik@uottawa.ca} \and Zhigang Tong}
\maketitle

\begin{abstract}
We consider regularly varying random vectors. Our goal is to estimate in a non-parametric way some characteristics related to conditioning on an extreme event, like the tail dependence coefficient. We introduce a quasi-spectral decomposition that allow to improve efficiency of estimators. Asymptotic normality of estimators is based on weak convergence of tail empirical processes.
Theoretical results are supported by simulation studies.
\end{abstract}


\section{Introduction}
Assume that $(X,Y)$ is a regularly varying random vector with index $\alpha$ and $F$ is the marginal distribution of $X$.
When dealing with extreme observations, we are often interested in estimating
\begin{align}\label{eq:to-estimate}
\esp\left[\psi\left(\frac{X}{x},\frac{Y}{x}\right) \mid (X,Y)\in xC\right]\;,
\end{align}
where $\psi:\Rset^2\to \Rset$, $C$ is a suitably chosen subset of $\uncompact{2}$ and $x$ is large. For example, $x$ can be chosen as $x=x_p=F^{\leftarrow}(1-p)$, where $p$ is small (The value $x_p$ is called in financial applications the \textit{Value-at-Risk}).
Special cases include estimation of the conditional tail distribution
\begin{align}\label{eq:to-estimate-conddistr}
1-G(y)=\lim_{x\to\infty}\pr(Y>y x\mid X>x)\;,
\end{align}
estimation of the conditional tail expectation (expected shortfall)
\begin{align*}
\lim_{x\to\infty}\esp\left[(Y/x) \mid X>x\right]\;,
\end{align*}
or extremal dependence measure
\begin{align*}
\lim_{x\to\infty}\esp\left[\frac{XY}{\|(X,Y)\|^2} \mid  \|(X,Y)\|>x\right]\;,
\end{align*}
where $\|\cdot\|$ is a vector norm on $\Rset^2$. The first problem is linked to
estimation of the tail dependence coefficient, the second one to modeling of the expected shortfall (\cite{cai:einmahl:haan:zhou:2012}), while the last one was introduced and studied in \cite{larsson:resnick:2009}.

In specific cases estimators of (\ref{eq:to-estimate}) can be obtained in a parametric or semi-parametric way and rely on a particular model chosen. Alternatively, one can consider nonparametric approaches (see \cite[Chapter 9]{beirlantetal:2004} for related theory and methods, as well as an extensive list of references). Specifically, having an i.i.d. sample
$(X_i,Y_i)$, $i=1,\ldots,n$, from $(X,Y)$,
estimation of the conditional tail distribution in (\ref{eq:to-estimate-conddistr}) can be achieved by
\begin{align}\label{eq:estimator-conddistr}
\frac{1}{k}\sum_{j=1}^n\1{\{Y_i> y X_{n:n-k},X_i>X_{n:n-k}\}}\;,
\end{align}
where $k$ is a deterministic sequence such that $k\to\infty$, $k/n\to 0$ and $X_{n:1}\leq \cdots\leq X_{n:n}$ are order statistics.
However, in order to provide reliable estimates of the conditional tail distribution one needs an appropriate number of pairs of observations such that the both components exceed the level $X_{n:n-k}$. This usually requires a very large number of observations. In summary, the estimator (\ref{eq:estimator-conddistr}) may not be particularly useful in practice.\\

We propose an alternative nonparametric approach to estimating the conditional tail distribution and more generally to estimating the expressions like the one in (\ref{eq:to-estimate}). The idea comes from \cite{basrak:segers:2009}, who considered regularly varying time series and defined a spectral and a tail spectral process. More specifically, in our context of bivariate vectors, regular variation implies that $(X/x,Y/x)$ conditionally on $X>x$ converges in distribution (when $x\to\infty$) to a random vector $(V_1\Theta_1,V_1 \Theta_2)$, where $V_1$ has a standard Pareto distribution, $\Theta_1$ is concentrated at $\{-1,1\}$, while $(\Theta_1,\Theta_2)$ is independent of $V_1$. Furthermore, $\Theta_2$ is a distributional limit of $Y/X$ given that $X>x$ and $x\to\infty$. The representation of the limiting vector is similar to the standard spectral decomposition (see \cite[Section 8.2.3]{beirlantetal:2004} or \cite[Section 6.1.2]{resnick:2007}), however, in our case the vector $(\Theta_1,\Theta_2)$ does not lie on a unit circle. Hence, we will call $(V_1,\Theta_1,\Theta_2)$ the \textit{quasi-spectral} decomposition.

As a consequence, if we assume for simplicity that all random variables are nonnegative, then the conditional tail distribution   can be expressed in terms of $\Theta_2$ as
\begin{align*}
\lim_{x\to\infty}\pr(Y> yx\mid X>x)=\esp\left[\left(\frac{\Theta_2}{y}\wedge 1\right)^{\alpha}\right]=\lim_{x\to\infty}
\esp\left[\left(\frac{Y}{yX}\wedge 1\right)^{\alpha}\mid X>x\right]\; .
\end{align*}
Thus, the estimator (\ref{eq:estimator-conddistr}) can be replaced with
\begin{align}\label{eq:estimator-conddistr-viasp}
\frac{1}{k}\sum_{j=1}^n\left(\frac{Y_j}{yX_j}\right)^{\alpha}\1{\{X_i>X_{n:n-k}\}}\;.
\end{align}
We will argue below that the estimator (\ref{eq:estimator-conddistr-viasp}) is more efficient than the one in (\ref{eq:estimator-conddistr}) (see also \cite{drees:segers:warchol:2014} in a different context of time series). Of course, if $\alpha$ is unknown, it needs to be replaced with its estimator, however, we will provide conditions that guarantee that estimation of $\alpha$ does not influence the limiting behaviour of the estimator of the conditional tail distribution. This observation will be also confirmed by simulation studies. Also, we note that the bivariate case can be easily extended to a general multivariate situation, still requiring only one component to be large.

Furthermore, the quasi-spectral decomposition can be useful in approximating the expected shortfall. It turns out that
\begin{align*}
\lim_{x\to\infty}x^{-1}\esp[Y\mid X>x]=\esp[V_1]\esp[\Theta_2]=\frac{\alpha}{\alpha-1}\lim_{x\to\infty}\esp\left[(Y/X)\mid X>x\right]
\end{align*}
whenever $\alpha>1$. Using the above identity we can construct two estimators of the expected shortfall. Asymptotic normality of an estimator that is based on the left-hand side of the above expression requires finiteness of the second moment, while an estimator motivated by the quasi-spectral representation on the right-hand side may have finite variance even when $\alpha\in (1,2)$.\\

In summary, the proposed estimation procedure based on the quasi-spectral representation may lead to improvement in terms of efficiency or in terms of the conditions required to achieve asymptotic normality, as compared to other nonparametric methods.\\

In order to support our statement, we proceed as follows.
In Section \ref{sec:prel} we recall the concept of multivariate regular variation (see \cite{resnick:2007}), followed by the quasi-spectral decomposition (Section \ref{sec:quasi-spectral}). We link it to the conditional tail distribution (Section \ref{sec:conddistr}) and the conditional tail expectation (Section \ref{sec:cte-1}).
We note that we present that section in a general framework of $d$-dimensional vectors. In Section \ref{sec:weakconv} we consider weak convergence of tail empirical processes based on deterministic and random levels. The theory is used to construct estimators of (\ref{eq:to-estimate}). Furthermore, some of the results in \cite{larsson:resnick:2009} and \cite{cai:einmahl:haan:zhou:2012} can be concluded from ours. The specific cases of the conditional tail distribution and the conditional tail expectation are discussed in Sections \ref{sec:conddistr-1} and \ref{sec:cte}, respectively. In the latter section we link our results to the estimation procedure in \cite{cai:einmahl:haan:zhou:2012}. In Section \ref{sec:sim} we conduct extensive simulation studies that show usefulness of our approach, while in the following one we apply our procedure to estimation of the tail dependence coefficient for some real data. Some technical details of proofs can be found in Section \ref{sec:technical-details}.
We finish our paper by addressing several technical issues like different marginals and directions of future research.

\section{Preliminaries}\label{sec:prel}
We start with some notation that will be used throughout the paper.
Unless otherwise stated, by $\vectorbold{y}$ we denote a vector $(y_1,\ldots,y_d)$.
For a vector $\vectorbold{y}$ we write $(\vectorbold{y},\vectorbold{\infty}]=(y_1,\infty]\times\cdots\times (y_d,\infty]$. For $C\subseteq \Rset^d$ and $y>0$ we denote $yC=\{y \vectorbold{x}: \vectorbold{x}\in C\}$.
As usual, for a given distribution $F$, we write $\tail{F}(x)=1-F(x)$.
\subsection{Multivariate regular variation}
We start with the following definition (see e.g. \cite[Theorem 6.1]{resnick:2007}).
\begin{definition}
  \label{def:mrv}
  A vector $\vectorbold{X}=(X_1,\ldots,X_d)$ in $\Rset^{d}$ is (multivariate)
  regularly varying
  if there exists a non zero Radon measure $\numult{X}$ on $\uncompact{d}$,
  called the exponent measure of $\vectorbold{X}$, such that
  $\numult{X}(\overline{\Rset}^{d}\setminus\Rset^{d})=0$ and a scaling sequence
  $\{\scalingseq\}$ such that the measure
  $n\pr(\scalingseq^{-1}\vectorbold{X}\in \cdot)$ converges vaguely on
  $\uncompact{d}$ to~the measure~$\numult{X}$, i.e.
\begin{align}
  \label{eq:def-regvar-mult}
  n \pr(\scalingseq^{-1}\vectorbold{X}\in \cdot) \convvague \numult{X} \; , \ \mbox{
    on $\uncompact{d}$} \; .
\end{align}
\end{definition}
The limiting measure is homogeneous with some index $-\alpha$, that is
$\numult{X}(yC)=y^{-\alpha}\numult{X}(C)$ for any $y>0$ and a relatively compact set $C$.
We call $-\alpha$ \textit{the index of regular variation} of $\vectorbold{X}$.

In what follows, we will assume that all components $X_i$ have the same distribution $F$ (see also
Section \ref{sec:technical-details} for extensions)
 and are nonnegative (the latter assumption is purely technical and can be easily relaxed). Then
\begin{align*}
  \lim_{x\to\infty} \frac{\pr(x^{-1}\vectorbold{X}\in A)}{\tail{F}(x)} =
  \frac{\numult{X}(A)} {\numult{X}(\{\vectorbold{x}:|x_1|>1\})} \equiv\numult{}(A) \; .
\end{align*}

\subsection{Quasi-spectral decomposition}\label{sec:quasi-spectral}
We can link vague convergence to weak convergence of
conditional probabilities. In particular, for relatively compact sets $A$, $B$ in $\uncompact{i-1}$, $\uncompact{d-i}$,
\begin{align*}
  \lim_{x\to\infty} \pr(x^{-1}\vectorbold{X}\in A\times (y,\infty]\times B\mid X_i>x)
  =\frac{\numult{X} \left(A\times (y,\infty]\times B\right) }{\numult{X}
    (\overline{\Rset}^{i-1}\times (1,\infty]\times\overline{\Rset}^{d-i})} \; .
\end{align*}
In this spirit, regular variation implies a \textit{quasi-spectral decomposition}. In time series context this approach was used in~\cite{basrak:segers:2009}.
\begin{proposition}\label{prop:quasi-spectral}
Let $\vectorbold{X}$ be a regularly varying random vector with non-negative regularly varying components with index $-\alpha$. Then conditionally on $X_1>x$, as $x\to\infty$
\begin{align*}
x^{-1}(X_1,\ldots,X_d)\;, \qquad \left(\frac{X_1}{x},\frac{X_2}{X_1},\ldots,\frac{X_d}{X_1}\right)
\end{align*}
converge in distribution to $(V_1,\ldots,V_d)$  and $(V_1,\Theta_2,\ldots,\Theta_d)$, where
\begin{enumerate}
\item $V_1$ has the Pareto distribution with index $-\alpha$;
\item $\Theta_j=V_j/V_1$, $j=2,\ldots,d$ and $(\Theta_2,\ldots,\Theta_d)$ is independent of $V_1$.
\end{enumerate}
\end{proposition}
\begin{proof}
A proof is given in Section~\ref{sec:proof-of-quasispectral}
\end{proof}
\begin{remark}{\rm
 Throughout the paper the quasi spectral-decomposition into $V_1$ and $(\Theta_2,\ldots,\Theta_d)$ is obtained by conditioning on $X_1$. We can condition on $X_{j}$ for any $j$. Note however that for each different $j$ we get different vectors $\vectorbold{V}$ (that depend formally on $j$).
}
\end{remark}
\subsection{Representation of conditional tail distribution}\label{sec:conddistr}
We use the quasi-spectral representation to express the conditional tail distribution.
\begin{corollary}\label{cor:lim-cond-prob}
Let $\vectorbold{X}$ be a regularly varying random vector with non-negative regularly varying components with index $-\alpha$. Then for $j_2,\ldots,j_l,j_{l+1},\ldots,j_d\in \{1,\ldots,d\}$ and $y_j>0$ we have
\begin{align}\label{eq:lim-cond-prob-via-quasisp}
&\lim_{x\to\infty}\pr(X_{j_{l+1}}>y_{j_{l+1}}x,\ldots,X_{j_d}>y_{j_d}x \mid X_1>x,X_{j_{2}}>x,\ldots,X_{j_{l}}>x)
\nonumber\\
&=
\frac{\esp\left[\left(\frac{\Theta_{j_{l+1}}}{y_{j_{l+1}}}\right)^{\alpha} \wedge \cdots\wedge
\left(\frac{\Theta_{j_{d}}}{y_{j_{d}}}\right)^{\alpha} \wedge \Theta_{j_{2}}^{\alpha}\wedge\cdots \Theta_{j_{l}}^{\alpha} \wedge 1  \right]}{\esp\left[ \Theta_{j_{2}}^{\alpha}\wedge\cdots \Theta_{j_{l}}^{\alpha} \wedge 1  \right]}\; .
\end{align}
\end{corollary}
\begin{proof}
Proposition \ref{prop:quasi-spectral} implies that for $y_1\geq 1$, $y_2,\ldots,y_d>0$,
\begin{align*}
&\lim_{x\to\infty}\pr(X_{1}>y_{1}x,\ldots,X_{d}>y_{d}x \mid X_1>x)= \pr(V_1>y_1,\ldots,V_d>y_d)\\
&=\pr(V_1>y_1,V_1\Theta_2>y_2,\ldots,V_1\Theta_d>y_d)\\
&=\alpha\int_{y_1\vee 1}^{\infty}\pr(\Theta_2>y_2/u,\ldots,\Theta_d>y_d/u)u^{-\alpha-1}\rmd u\\
&= \alpha\int_{y_1\vee 1}^{\infty}
\pr\left(\left(\frac{\Theta_2}{y_2}\right)^{\alpha}>u^{-\alpha},\ldots,\left(\frac{\Theta_d}{y_d}\right)^{\alpha}>u^{-\alpha}\right)u^{-\alpha-1}\rmd u \\
&=\esp\left[\left(\frac{1}{y_1}\wedge\frac{\Theta_2}{y_2}\wedge \cdots\frac{\Theta_d}{y_d}\right)^{\alpha}\right]\; .
\end{align*}
Furthermore,
\begin{align*}
&\pr(X_{j_{l+1}}>y_{j_{l+1}}x,\ldots,X_{j_d}>y_{j_d}x \mid X_1>x, X_{j_{2}}>x,\ldots,X_{j_{l}}>x) =\\
&=\frac{\pr(X_{j_{2}}>x,\ldots,X_{j_{l}}>x, X_{j_{l+1}}>y_{j_{l+1}}x,\ldots,X_{j_d}>y_{j_d}x\mid X_1>x)}
{\pr(X_{j_{2}}>x,\ldots,X_{j_{l}}>x\mid X_1>x)}
\end{align*}
and the result follows.
\end{proof}
We note that the numerator and the denumerator in (\ref{eq:lim-cond-prob-via-quasisp}) can be expressed as limits. In particular, via Proposition \ref{prop:quasi-spectral}, the numerator in (\ref{eq:lim-cond-prob-via-quasisp}) equals
\begin{align*}
\lim_{x\to\infty}\esp\left[g\left(\frac{X_{j_2}}{X_1},\ldots,\frac{X_{j_l}}{X_1},\frac{X_{j_{l+1}}}{y_{j_{l+1}}X_1},\ldots,
\frac{X_{j_{d}}}{y_{j_{d}}X_1}\right)\mid X_1>x\right]\;
\end{align*}
with a bounded and continuous function
$g(u_2,\ldots,u_d)=(u_2\wedge \cdots \wedge u_d\wedge 1)^{\alpha}$.
Consequently, for $y>0$, and setting $(X_1,X_2)=(X,Y)$
\begin{align}\label{eq:lcp-quasisp-d=2}
\lim_{x\to\infty}\pr(Y>y x \mid X>x)= \esp\left[\left(\frac{\Theta_2}{y}\wedge 1\right)^{\alpha}\right]=\lim_{x\to\infty}\esp\left[\left(\frac{Y}{yX}\wedge 1\right)^{\alpha}\mid X>x\right]\; .
\end{align}

\subsection{Representation of conditional tail expectation}\label{sec:cte-1}
\begin{corollary}\label{cor:lim-cond-exp}Let $\vectorbold{X}$ be a regularly varying random vector with non-negative regularly varying components with index $-\alpha$.
Assume moreover that for some $\delta>0$ we have
\begin{align}\label{eq:moment-bound}
\sup_{x>0}\esp\left[\left(\frac{X_{j_d}}{x}\right)^{1+\delta} \mid X_1>x, X_{j_{2}}>x,\ldots,X_{j_{l}}>x\right]<\infty\; .
\end{align}
Then
\begin{align*}
\esp\left[\frac{X_{j_d}}{x} \mid X_1>x, X_{j_{2}}>x,\ldots,X_{j_{l}}>x\right]=\frac{\alpha}{\alpha-1}\frac{\esp\left[\Theta_{j_d} \left(\Theta_{j_2}\wedge \cdots \wedge \Theta_{j_l}\wedge 1\right)^{\alpha-1} \right]\; .
}{\esp[\Theta_{j_1}^\alpha\wedge \ldots \wedge \Theta_{j_l}^\alpha\wedge 1]}.
\end{align*}
\end{corollary}
\begin{proof}
We note first that (\ref{eq:moment-bound}) implies that $\alpha>1$.
Let $A\subseteq (0,\infty)$. Proposition~\ref{prop:quasi-spectral} implies that as $x\to\infty$
\begin{align*}
\esp\left[\frac{X_{j_d}}{x}\1{\{X_{j_d}\leq xA\}} \mid X_1>x, X_{j_{2}}>x,\ldots,X_{j_{l}}>x\right]\to\frac{\esp[V_{j_d}\1{\{1<V_{j_d}\leq A\}} \1{\{ V_1>1,V_{j_2}>1,\ldots, V_{j_l}>1\}}]}{\esp[\Theta_{j_1}^\alpha\wedge \ldots \wedge \Theta_{j_l}^\alpha\wedge 1]}\; .
\end{align*}
A computation similar to Corollary \ref{cor:lim-cond-prob} yields that
the numerator in the last expression is
\begin{align*}
\frac{\alpha}{\alpha-1}\esp\left[\Theta_{j_d} \left(\Theta_{j_2}\wedge \cdots \wedge \Theta_{j_l}\wedge 1\right)^{\alpha-1} \right]\; .
\end{align*}
Furthermore, (\ref{eq:moment-bound}) implies
\begin{align*}
&\lim_{A\to\infty}\limsup_{x\to\infty}\esp\left[\frac{X_{j_d}}{x}\1{\{X_{j_d}> xA\}} \mid X_1>x, X_{j_{2}}>x,\ldots,X_{j_{l}}>x\right]\\
&\leq \lim_{A\to\infty}
A^{-\delta}\limsup_{x\to\infty}\esp\left[\left(\frac{X_{j_d}}{x}\right)^{1+\delta} \mid X_1>x, X_{j_{2}}>x,\ldots,X_{j_{l}}>x\right]=0\; .
\end{align*}
\end{proof}
In particular, if $\alpha>1$ then setting again $(X_1,X_2)=(X,Y)$,
\begin{align}\label{eq:cte}
\lim_{x\to\infty}\esp\left[\frac{Y}{x}\mid X>x\right]=\frac{\alpha}{\alpha-1}\esp\left[\Theta_2\right]
=\frac{\alpha}{\alpha-1}\lim_{x\to\infty}\esp\left[\frac{Y}{X}\mid X>x\right]=:\constant_{\rm CTE}\;
\end{align}
and the limit is strictly positive in case of extremal dependence, that is when the limiting exponent measure $\numult{X}$ in (\ref{eq:def-regvar-mult}) is not concentrated on the axes.

\section{Weak convergence of tail empirical process}\label{sec:weakconv}
For clarity of notation we consider the case $d=2$ and a vector $(X_1,X_2)$ is written as $(X,Y)$. Recall that all random variables are non-negative with the distribution function $F$ and regularly varying with the same index $-\alpha$. Assume that we have an i.i.d. sample $(X_{j},Y_{j})$, $j=1,\ldots,n$, from the distribution of $(X,Y)$. Let $\psi:\Rset^2\to \Rset_+$.
In what follows $u_n$ denotes a scaling sequence, that is the sequence such that $u_n\to\infty$ and $n\tail{F}(u_n)\to \infty$. For $s_0>0$, define the tail empirical function
\begin{align}\label{eq:ted}
\tilde T_n(s;\psi,C)= \frac{1}{n\tail{F}_X(u_n)}\sum_{j=1}^n \psi\left(\frac{X_{j}}{u_n},\frac{Y_{j}}{u_n}\right) \1{\{(X_{j},Y_{j})\in su_n C\}}\; ,\qquad s\geq s_0\; ,
\end{align}
and $T_n(s;\psi,C)=\esp[\tilde T_n(s;\psi,C)]$. If $\psi$ is homogeneous with index $\gamma$ then Lemma \ref{lem:cond-exp} implies
\begin{align}\label{eq:homogeneity-scaling}
T(s;C,\psi)\equiv \lim_{n\to\infty}T_n(s;C,\psi)=s^{\gamma-\alpha}\int_C \psi(v_1,v_2)\numult{}(\rmd v_1,\rmd v_2)\; ,
\end{align}
whenever $\psi$ satisfies the appropriate integrability condition (see (\ref{eq:condition-on-psi}) below).

Consider the tail empirical process
\begin{align}\label{eq:tep}
\process[G]_n(s;\psi,C)=\sqrt{n\tail{F}(u_n)}\left\{\tilde T_n(s;\psi,C)-T_n(s;\psi,C)\right\}\;.
\end{align}
Also, define $\process[G]_n^*(\cdot)$ to be the process $\process[G]_n(\cdot;\psi,C)$ for the function $\psi\equiv 1$ and the set
$C=\{(x_1,x_2):x_1>1\}$.

The main result of this section is the following weak convergence for the tail empirical function. A proof is given in Section \ref{sec:technical-details}.
\begin{theorem}\label{theo:weakconv}
Let $s_0>0$.
Assume that $(X_{j},Y_{j})$ are i.i.d. regularly varying random vectors with non-negative regularly varying components with index $-\alpha$. If moreover
\begin{enumerate}
\item $u_n\to\infty$ and $n\tail{F}(u_n)\to \infty$;
\item The function $\psi$ is homogenous with order $\gamma\in \Rset$;
\item For $0<s_0\leq s\leq t$ we have $tC\subseteq sC$;
\item There exists $\delta>0$ such that $\int_C \psi^{2+\delta}(v_1,v_2)\numult{}(\rmd v_1,\rmd v_2)<\infty$;
\end{enumerate}
then
\begin{align}\label{eq:tep-weakconv}
(\process[G]_n^*(\cdot),\process[G]_n(\cdot;\psi,C))\convweak (\process[G]^*(\cdot),\process[G](\cdot;\psi,C))
\end{align}
in $\spaceD([s_0,\infty))\times \spaceD( [s_0,\infty))$, where $\process[G]^*(\cdot)$, $\process[G](\cdot;\psi,C)$ are Gaussian processes with the covariance functions
\begin{align*}
\cov(\process[G]^*(s),\process[G]^*(t))=(s\vee t)^{-\alpha}\;,
\end{align*}
\begin{align*}
\cov (\process[G](s;\psi,C),\process[G](t;\psi,C))=(s\vee t)^{2\gamma-\alpha}\int_C\psi^{2}(v_1,v_2)\numult{}(\rmd v_1,\rmd v_2) \; .
\end{align*}
\end{theorem}


\subsection{Tail empirical process with random levels}
To apply the weak convergence established in Theorem~\ref{theo:weakconv} one needs to choose $u_n$. The sequence depends on the marginal distribution which is unknown. Hence, we consider the tail empirical process with random levels. We refer the reader to \cite{rootzen:2009} and \cite{kulik:soulier:2011}.

The second issue is that the centering in the tail empirical process (\ref{eq:tep}) is $T_n(s;\psi,C)$ not its limit $T(s;\psi,C)$. This will be handled by an appropriate "no-bias" condition.

To proceed, choose a sequence $k=k_n$ such that $k\to\infty$ and $k/n\to 0$ and define $u_n$ by $k=n\tail{F}(u_n)$.
Let $X_{n:1}\leq X_{n:2}\leq \cdots \leq X_{n:n}$ be order statistics from $X_{j}$, $j=1,\ldots,n$. First, from Theorem \ref{theo:weakconv} we conclude the following weak convergence. Let $T_n(s)=\tail{F}(su_n)/\tail{F}(u_n)$.
\begin{corollary}\label{cor:order-stats}
Assume that the conditions of Theorem~\ref{theo:weakconv} are satisfied. Furthermore, assume that
  the distribution function $F$ is continuous
  and that
  \begin{align}
    \label{eq:tep-derivee-uniform-convergence}
    \lim_{n\to\infty} \TEDlimit_n'(s) = -\alpha s^{-\alpha-1} \; ,
  \end{align}
  uniformly in a neighborhood of~1.
Then
\begin{align*}
\left(\sqrt{k}\left\{\frac{X_{n:n-k}}{u_n}-1\right\},\process[G]_n(\cdot;\psi,C)\right)\convweak (\alpha^{-1}\process[G]^*(1),\process[G](\cdot;\psi,C))\; .
\end{align*}
\end{corollary}
We note that the normal convergence of the order statistics is standard (see e.g. \cite[Theorem 2.4.1]{dehaan:ferreira:2006}), but we need to argue that the convergence holds jointly.

Furthermore, we impose the following no-bias condition:
\begin{align}\label{eq:bias}
\lim_{n\to\infty}\sup_{s>s_0}\sqrt{k}|T_n(s;\psi,C)-T(s;\psi,C)|=0\;.
\end{align}
This leads to the following empirical processes
\begin{align*}
\hat{\process[G]}_n(s;\psi,C)=\sqrt{k}\left\{\hat T_n(s;\psi,C) - T_n(s;\psi,C)\right\}\;
\end{align*}
\begin{align*}
\hat{\hat{\process[G]}}_n(s;\psi,C)=\sqrt{k}\left\{\hat{\hat T}_n(s;\psi,C) - T_n(s;\psi,C)\right\}\; ,
\end{align*}
where
\begin{align}\label{eq:ted-random}
  \hat T_n(s;\psi,C) & =  \frac{1}{k}\sum_{j=1}^n \psi
    \left( \frac{X_{j}}{\tepseq},\frac{Y_{j}}{\tepseq}\right) \1{\{(X_{j},Y_{j})\in s X_{n:n-k} C\}} \;
\end{align}
and
\begin{align*}
  \hat{\hat T}_n(s;\psi,C) & =  \frac{1}{k}\sum_{j=1}^n \psi
    \left( \frac{X_{j}}{X_{n:n-k}},\frac{Y_{j}}{X_{n:n-k}}\right) \1{\{(X_{j},Y_{j})\in s X_{n:n-k} C\}} \; .
\end{align*}
\begin{theorem}
  \label{theo:fclt-random-ext}
  Assume that the conditions of Theorem~\ref{theo:weakconv} are satisfied. Furthermore,
  the distribution function $F$ is continuous and
  (\ref{eq:tep-derivee-uniform-convergence}) holds.
   Then
  \begin{align}
  \label{eq:fclt-random-ext}
  \hat{\process[G]}_n(s;\psi,C) \convweak \process[G](s;\psi,C)  +s^{\gamma-\alpha}\frac{1}{\alpha}T'(1;\psi,C)\process[G]^*(1)
  \; ,
  \end{align}
and
\begin{align}\label{eq:fclt-random-random-ext}
\hat{\hat{\process[G]}}_n(s;\psi,C) \convweak \process[G](s;\psi,C)
+\frac{1}{\alpha}s^{\gamma-\alpha}T'(1;\psi,C)\process[G]^*(1)-\frac{\gamma}{\alpha}s^{\gamma-\alpha}T(1;\psi,C)\process[G]^*(1)\;.
\end{align}
in $\spaceD([s_0,\infty))$. If moreover (\ref{eq:bias}) is satisfied, then the centering $T_n(s;\psi,C)$ can be replaced with its limit $T(s;\psi,C)$.
\end{theorem}

\section{Conditional tail distribution}\label{sec:conddistr-1}
If we choose $\psi\equiv 1$ and $C=\{(x_1,x_2):x_1>1,x_2>y\}$, $y>0$, then $\tilde T_n(s;\psi,C)$ in (\ref{eq:ted}) becomes
\begin{align}\label{eq:tef-1}
\tilde T_n^{(1)}(s;y)= \frac{1}{n\tail{F}(u_n)}\sum_{j=1}^n \1{\{X_{j}>su_n,Y_{j}>su_n y\}}\; .
\end{align}
Furthermore,
\begin{align*}
T_n^{(1)}(s;y)=\frac{\pr(X>su_n,Y> su_n y)}{\pr(X>u_n)}\; , \; ,
T^{(1)}(s;y)=s^{-\alpha}\int_{(1,\infty]\times(y, \infty]} \numult{}(\rmd v_1,\rmd v_2)\; .
\end{align*}
Hence, $T^{(1)}(1;y)=\lim_{n\to\infty}\pr(Y>u_n y\mid X>u_n)$ is the limiting conditional tail distribution and $T^{(1)}(1;1)$ is the tail dependence coefficient.
We note that in terms of the quasi-spectral representation the limiting variance is
\begin{align}\label{eq:limiting-variance-lcp}
s^{-\alpha}\esp\left[\left(\frac{\Theta_2}{y}\wedge 1\right)^{\alpha}\right]\; .
\end{align}
If we choose $\psi(x_1,x_2)=(x_2/(yx_1)\wedge 1)^{\alpha}$ and $C=\{(x_1,x_2):x_1>1\}$ then
\begin{align}\label{eq:tef-2}
\tilde T_n^{(2)}(s;y)= \frac{1}{n\tail{F}(u_n)}\sum_{j=1}^n \left(\frac{Y_{j}}{yX_{j}}\wedge 1\right)^\alpha\1{\{X_{j}>su_n\}}\;,
\end{align}
\begin{align*}
T_n^{(2)}(s;y)=\frac{1}{\tail{F}(u_n)}\esp\left[\left(\frac{Y}{yX}\wedge 1\right)^\alpha\1{\{X>su_n\}}\right]\; ,\qquad
T^{(2)}(s;y)=\lim_{n\to\infty}T_n^{(2)}(s;y)\; .
\end{align*}
In particular, using (\ref{eq:lcp-quasisp-d=2}),
\begin{align*}
T^{(2)}(1;y)=\esp\left[\left(\frac{\Theta_2}{y}\wedge 1\right)^\alpha\right]=\lim_{n\to\infty}\pr(Y>u_n y\mid X>u_n)\; .
\end{align*}
Theorem \ref{theo:weakconv} implies that $\sqrt{n \tail{F}(u_n)}\left\{\tilde T^{(2)}(s;y)-T_n^{(2)}(s;y)\right\}$ converges to a Gaussian process $\process[G]^{(2)}(s;y)$ with the limiting variance
\begin{align}\label{eq:limiting-variance-lcp-quasispectral}
s^{-\alpha}\esp\left[\left(\frac{\Theta_2}{y}\wedge 1\right)^{2\alpha}\right]\; .
\end{align}
which is smaller than the one given in (\ref{eq:limiting-variance-lcp}) whenever $y\geq 1$.

Hence, both tail empirical functions in (\ref{eq:tef-1}) and (\ref{eq:tef-2}) can be used to construct estimators of the  limiting conditional tail distribution. Specifically, we can use
\begin{align}\label{eq:tef-1-estim}
\hat T_n^{(1)}(1;y)= \frac{1}{k}\sum_{j=1}^n \1{\{Y_j>yX_{n:n-k},X_{j}>X_{n:n-k}\}}\;,
\end{align}
\begin{align}\label{eq:tef-2-estim}
\hat T_n^{(2)}(1;y)= \frac{1}{k}\sum_{j=1}^n \left(\frac{Y_{j}}{yX_{j}}\wedge 1\right)^\alpha\1{\{X_{j}>X_{n:n-k}\}}\;,
\end{align}
the latter one when $\alpha$ is known. The above discussion indicates that the second estimator can be asymptotically more efficient than the first one.
\subsection{Unknown $\alpha$}\label{sec:lcp-alpha}
Let $\hat\alpha$ be an estimator of $\alpha$. We redefine $\hat T_n^{(2)}(1;y)$ from (\ref{eq:tef-2-estim}) as
\begin{align}\label{eq:tef-2-estim-estim}
\hat T_n^{(2),\hat\alpha}(1;y)= \frac{1}{k}\sum_{j=1}^n \left(\frac{Y_{j}}{yX_{j}}\wedge 1\right)^{\hat\alpha}\1{\{X_{j}>X_{n:n-k}\}}\;.
\end{align}
We have
\begin{align*}
&\sqrt{k}\left\{\frac{1}{k}\sum_{j=1}^n \left(\frac{Y_{j}}{yX_{j}}\wedge 1\right)^{\hat{\alpha}}\1{\{X_{j}>X_{n:n-k}s\}}-T^{(2)}(1;y)\right\}\\
&=\sqrt{k}\left\{\frac{1}{k}\sum_{j=1}^n \left(\frac{Y_{j}}{yX_{j}}\wedge 1\right)^{\hat{\alpha}}\1{\{X_{j}>X_{n:n-k}\}}-\frac{1}{k}\sum_{j=1}^n \left(\frac{Y_{j}}{yX_{j}}\wedge 1\right)^\alpha\1{\{X_{j}>X_{n:n-k}\}}\right\}\\
&\quad +\sqrt{k}\left\{\frac{1}{k}\sum_{j=1}^n \left(\frac{Y_{j}}{yX_{j}}\wedge 1\right)^\alpha\1{\{X_{j}>X_{n:n-k}\}}-T^{(2)}(1;y)\right\}=U_1+U_2(y)\; .
\end{align*}
We already know (cf. (\ref{eq:fclt-random-ext})) that
\begin{align*}
U_2(y)
\Rightarrow \process[G]^{(2)}(1;y)+\alpha^{-1} (T^{(2)})'(1;y)\process[G]^*(1)\; .
\end{align*}
Using the first order Taylor expansion for $\alpha\rightarrow z^\alpha$, we have
\begin{align*}
U_1&\approx\sqrt{k}\left\{O_P(\hat{\alpha}-\alpha)\frac{1}{k}\sum_{j=1}^n \left(\frac{Y_{j}}{yX_{j}}\wedge 1\right)^\alpha\log\left(\frac{Y_{j}}{yX_{j}}\wedge 1\right)\1{\{X_{j}>X_{n:n-k}\}}\right\}.
\end{align*}
Let $k=o(\tilde k)$ and $\hat\alpha_{\tilde k}$ be the Hill estimator. We know that $\sqrt{\tilde k}(\hat\alpha_{\tilde k}-\alpha)$ converges to a normal random variable. Hence, in order to show that $U_1$ is of a smaller order than $U_2(y)$ it suffices to justify that
\begin{align*}
\frac{1}{k}\sum_{j=1}^n \left(\frac{Y_{j}}{yX_{j}}\wedge 1\right)^\alpha\log\left(\frac{Y_{j}}{yX_{j}}\wedge 1\right)\1{\{X_{j}>X_{n:n-k}\}}\;
\end{align*}
is bounded in probability,uniformly in $y$.
Assume that for $\delta>0$ we have
\begin{align*}
\esp\left[(\Theta_2\wedge 1)^{\alpha+\delta}|\log(\Theta_2\wedge 1)|^{1+\delta}\right]<\infty \; .
\end{align*}
Then recalling that $k=n\tail{F}(u_n)$ and $X_{n:n-k}/u_n\convprob 1$,
\begin{align*}
&\limsup_{n\to\infty}\esp\left[\frac{1}{k}\sum_{j=1}^n \left(\frac{Y_{j}}{yX_{j}}\wedge 1\right)^\alpha\left|\log\left(\frac{Y_{j}}{yX_{j}}\wedge 1\right)\right|\1{\{X_{j}>X_{n:n-k}\}}\right]
\leq
\esp\left[(\Theta_2\wedge 1)^{\alpha}|\log(\Theta_2\wedge 1)|\right]\; .
\end{align*}
Hence, $U_1$ is negligible and there is no effect of estimation of $\alpha$.

\section{Conditional Tail Expectation}\label{sec:cte}
If we choose $\psi(x_1,x_2)=x_2$ and $C=\{(x_1,x_2):x_1>1\}$ then $\tilde T_n(s;\psi,C)$ in (\ref{eq:ted}) becomes
\begin{align}\label{eq:ted-cte}
\tilde T_n^{(3)}(s)= \frac{1}{n\tail{F}(u_n)}\sum_{j=1}^n \frac{Y_{j}}{u_n}\1{\{X_{j}>su_n\}}\; ,
\end{align}
\begin{align*}
T_n^{(3)}(s)=\frac{1}{\tail{F}(u_n)}\esp\left[\frac{Y}{u_n}\1{\{X_{1}>su_n\}}\right]\; ,\;,
T^{(3)}(s)=s^{1-\alpha}\int_{(1,\infty]\times(0, \infty]} v_2\numult{}(\rmd v_1,\rmd v_2)\; .
\end{align*}
We note that the limiting variance can be represented as
\begin{align}\label{eq:limiting-variance-cte}
s^{2-\alpha}\frac{\alpha}{\alpha-2}\esp[\Theta_2^2]\; .
\end{align}
If we choose $\psi(x_1,x_2)=\frac{\alpha}{\alpha-1}\frac{x_2}{x_1}$ and $C=\{(x_1,x_2):x_1>1\}$ then
\begin{align}\label{eq:ted-cte-viaspectral}
\tilde T_n^{(4)}(s)= \frac{1}{n\tail{F}(u_n)}\frac{\alpha}{\alpha-1}\sum_{j=1}^n \frac{Y_{j}}{X_{j}}\1{\{X_{j}>su_n\}}\; ,
\end{align}
\begin{align*}
T^{(4)}(s)=\lim_{n\to\infty}\esp[\tilde T_n^{(4)}(s)]=s^{-\alpha}\frac{\alpha}{\alpha-1}\int_{(1,\infty]\times(0, \infty]} \frac{v_2}{v_1}\numult{}(\rmd v_1,\rmd v_2)\; .
\end{align*}
In particular, by (\ref{eq:cte})
\begin{align}\label{eq:cte-viaspectral}
T^{(4)}(1)=\frac{\alpha}{\alpha-1}\lim_{n\to\infty}\esp\left[\frac{Y}{X}\mid X>u_n\right]=\lim_{n\to\infty}\esp\left[\frac{Y}{u_n}\mid X>u_n\right]\; .
\end{align}
We have furthermore
\begin{align}
\var(\process[G]^{(4)}(s))&=s^{-\alpha}\left(\frac{\alpha}{\alpha-1}\right)^2\int_{(0,\infty)}v_2^2\int_1^\infty \frac{1}{v_1^2}\numult{}(\rmd v_1,\rmd v_2)\label{eq:variance-cte-viaspectral-1}\\
&\leq s^{-\alpha}\left(\frac{\alpha}{\alpha-1}\right)^2\int_{(0,\infty)}v_2^2\int_1^\infty \numult{}(\rmd v_1,\rmd v_2)\; .\label{eq:variance-cte-viaspectral-2}
\end{align}
The integral in (\ref{eq:variance-cte-viaspectral-2}) is finite whenever $\alpha>2$. However, the integral in (\ref{eq:variance-cte-viaspectral-1}) may exists even when $\alpha<2$ (take trivially the situation of $Y=X$ or $Y=\phi X+\sigma |Z|$, where $\phi>0$, $X$ is regularly varying with index $-\alpha$ and support contained in $(\epsilon,\infty)$, $\epsilon>0$, independent of a standard normal random variable $Z$.)

The limiting variance can be written as
\begin{align}\label{eq:limiting-variance-cte-quasispectral}
s^{-\alpha}\left(\frac{\alpha}{\alpha-1}\right)^2\esp[\Theta_2^2]\; .
\end{align}
We note that for $s=1$ the limiting variance in (\ref{eq:limiting-variance-cte-quasispectral}) is smaller than the one in (\ref{eq:limiting-variance-cte}). Furthermore, the effect of estimating $\alpha$ is negligible if we use an estimator of $\alpha$ with a faster rate of convergence, as described in Section \ref{sec:lcp-alpha}.

\subsection{Modelling Conditional Tail Expectation}
Let $U(t) = F^\leftarrow(1-1/t)$ be the upper quantile function.
For a small $p\in (0,1)$ we have $\pr(X>U(1/p))=p$.
Our goal is to estimate
\begin{align*}
\estimp(p)=\esp\left[Y_1\mid X_1>U(1/p)\right]
\end{align*}
when $p$ is small. In case of extremal dependence we have (cf. (\ref{eq:cte})) whenever $p\to 0$,
\begin{align}\label{eq:approximation-of-thetap}
\estimp(p)\approx \constant_{\rm CTE} U_X(1/p)\; ,
\end{align}
where $\constant_{\rm CTE}=\lim_{x\to\infty}x^{-1}\esp[Y_1\mid X_1>x]\in (0,\infty)$.
If we model the tail by a generalized extreme value distribution, then $U(1/p)$ can be estimated using the representation (5.9) in \cite{beirlantetal:2004}, while $\hat\constant_{\rm CTE}$ can be
estimated
using the tail empirical functions (\ref{eq:ted-cte}) and (\ref{eq:ted-cte-viaspectral}) as follows. We take $s^{-1}\tilde T_n^{(3)}(s)$ and $\tilde T_n^{(4)}(s)$ and then replace $s$ with $X_{n:n-k}/u_n$ to obtain
\begin{align}
&\hat\constant_{\rm CTE}^{(3)}=\hat T_n^{(3)}(1)=\frac{1}{k}\sum_{j=1}^n \frac{Y_{j}}{X_{n:n-k}}\1{\{X_{j}>X_{n:n-k}\}}\label{eq:cte-estimators-1}\\
&\hat\constant_{\rm CTE}^{(4)}=\hat T_n^{(4)}(1)=\frac{1}{k}\frac{\hat\alpha}{\hat\alpha-1}\sum_{j=1}^n \frac{Y_{j}}{X_j}\1{\{X_{j}>X_{n:n-k}\}}\; .\label{eq:cte-estimators-2}
\end{align}
Then, $\hat\constant_{\rm CTE}$ can be
chosen to be one of the estimators defined in (\ref{eq:cte-estimators-1})-(\ref{eq:cte-estimators-2}).

Let now $X$ be regularly varying.
The function $U$ is regularly varying as $p\to 0$. If $\hat F_{n,X}$ is the empirical distribution function associated with $X_1,\ldots,X_n$ and we set
$\hat U_{n,X}(t)=\hat F_{n,X}(1-1/t)$, then $\hat U_n(n/k)=X_{n:n-k}$. Thus, when $n/k \approx 1/p$,
we have the following approximation (see \cite[p. 119]{beirlantetal:2004}):
\begin{align}\label{eq:approx}
\estimp(p)\approx \constant_{\rm CTE}U_X(1/p)\approx \constant_{\rm CTE}U_X(n/k)\left(\frac{k}{np}\right)^{1/\alpha}\; .
\end{align}
Hence,
we can estimate
\begin{align*}
\hat\estimp(p) =\hat\constant_{\rm CTE}  X_{n:n-k}\left(\frac{k}{np}\right)^{\hat\alpha^{-1}}\;,
\end{align*}
where $\hat\alpha$ is an estimator of $\alpha$.

Equation (\ref{eq:cte-estimators-1})
leads
to the following estimators of $\estimp(p)$:
\begin{align}\label{eq:modeling-cte-1}
\tilde\estimp^{(3)}(p) =\frac{1}{k}\sum_{j=1}^n Y_{j}\1{\{X_{j}>X_{n:n-k}\}}  \times\left(\frac{k}{np}\right)^{\alpha^{-1}}\;, \;
\hat\estimp^{(3)}(p) =\frac{1}{k}\sum_{j=1}^n Y_{j}\1{\{X_{j}>X_{n:n-k}\}}  \times\left(\frac{k}{np}\right)^{\hat\alpha^{-1}}\; .
\end{align}
We note that
(\ref{eq:modeling-cte-1}) is precisely the estimator used in \cite{cai:einmahl:haan:zhou:2012} and our Theorem \ref{theo:fclt-random-ext} can be used to conclude their Theorem 1 under slightly different conditions. Indeed, using (\ref{eq:approx}) and noting that
$U(n/k)=u_n$ we have
\begin{align}
&\sqrt{\statinterseq}\left\{\frac{\tilde\estimp^{(3)}(p)}{\estimp(p)}-1\right\}\approx
\sqrt{\statinterseq}\frac{1}{\constant_{\rm CTE}}
\left\{\frac{\hat\constant_{\rm CTE}^{(3)} \orderstat{n:n-\statinterseq}}{U_X(n/k)}-\constant_{\rm CTE}\right\} \nonumber\\ 
&\phantom{=}+\sqrt{\statinterseq} \left\{\frac{\left(\frac{k}{np}\right)^{1/\alpha}U_X(n/k)}{U_X(1/p)}-1\right\} \frac{\hat\constant_{\rm CTE}^{(3)} \orderstat{n:n-\statinterseq}}{\constant_{\rm CTE}U_X(n/k)}\; .
\label{eq:estimp-decomposition-cte-2}
\end{align}
We can recognize
$\hat\constant_{\rm CTE}^{(3)} \orderstat{n:n-\statinterseq}/U_X(n/k)$ to be
\begin{align*}
\frac{1}{k}\sum_{j=1}^n \frac{Y_j}{u_n}\1{\{X_j>X_{n:n-k}\}}
\end{align*}
and its convergence can be concluded from (\ref{eq:fclt-random-ext}), while the bias term in (\ref{eq:estimp-decomposition-cte-2}) can be handled by imposing a second order condition as in \cite{cai:einmahl:haan:zhou:2012}.

Now, the case of estimated $\alpha$ in $\hat\estimp^{(3)}(p)$. Applying the first order Taylor expansion, we have
\begin{align*}
\hat\estimp^{(3)}(p)\approx \estimp(p)
+\hat\constant_{\rm CTE}^{(3)}\orderstat{n:n-\statinterseq}
\left(\frac{1}{\hat\alpha}-\frac{1}{\alpha}\right)\left(\frac{k}{np}\right)^{1/\alpha} \log\left(\frac{k}{np}\right)\;,
\end{align*}
so that
\begin{align*}
&\sqrt{\statinterseq}\left\{\frac{\hat \estimp^{(3)}(p)}{ \estimp(p)}-1\right\}\approx
\sqrt{\statinterseq}\left\{\frac{\tilde \estimp^{(3)}(p)}{\estimp(p)}-1\right\}\\
&+\sqrt{\delta_n}\left(\frac{1}{\hat\alpha}-\frac{1}{\alpha}\right) \frac{\sqrt{\statinterseq}}{\sqrt{\delta_n}}\log\left(\frac{k}{np}\right)\frac{\hat\constant_{\rm CTE}^{(3)}\orderstat{n:n-\statinterseq}}{\estimp(p)}\left(\frac{k}{np}\right)^{1/\alpha} \; .
\end{align*}
If for some $\delta_n\to\infty$ and a random variable $\Delta$ we have
\begin{align*}
\sqrt{\delta_n}\left(\frac{1}{\hat\alpha}-\frac{1}{\alpha}\right) \convdistr \Delta
\end{align*}
and $\lim_{n\to\infty}\frac{\sqrt{\statinterseq}}{\sqrt{\delta_n}}\log\left(\frac{k}{np}\right)=r\in [0,\infty)$, then
estimation of $\alpha$ yields an additional contribution $r\Delta$. This is exactly the situation of Theorem 1 in \cite{cai:einmahl:haan:zhou:2012}, however note that they did not require that the vector $(X,Y)$ is regularly varying. Nevertheless, their Theorem 1 can be recovered from our results.

\section{Implementation. Simulation studies}\label{sec:sim}
We perform simulation studies to illustrate our theoretical results.
We illustrate estimation of the tail dependence coefficient
\begin{align*}
\mbox{\rm TDC}:=\lim_{x\to\infty}\pr(Y> x\mid X>x)\; .
\end{align*}
We use
the estimators $\hat T_n^{(1)}(1;1)$, $\hat T_n^{(2)}(1;1)$, $\hat T_n^{(2),\hat\alpha}(1;1)$ defined in (\ref{eq:tef-1-estim}),
(\ref{eq:tef-2-estim}), (\ref{eq:tef-2-estim-estim}).
At the first step we plot estimates computed for different numbers $k$ of order statistics. Next, we conduct Monte Carlo estimation for particular choices of $k$ (5\%, 10\%, 20\%, 30\% and 40\% of observations). Number of Monte Carlo iterations is chosen to be 1000.

Our simulations indicate that the quasi-spectral method is less variable more robust (in terms of the choice of $k$) than the standard empirical method, even if the parameter $\alpha$ has to be estimated.

\subsection{A toy example: simple linear model}
We simulate 1000 observations from the model $Y= \phi X +\sigma |Z|$, where $\phi\in (0,1)$, $\sigma>0$, $X$ is standard Pareto with $\alpha>0$ and $Z$ is standard normal. In this case the tail dependence coefficient is $\phi^{\alpha}$.

Figure 1 
shows
shows the estimated values using the three estimators, computed for different values of $k$, where $k$
is the number  of order statistics being used. On the $x$-axes actual values of order statistics $X_{n:1},\ldots,X_{n:n}$ are plotted in the increasing order.
Hence, the estimators computed at the left-end of each picture use a large number of order statistics, while at the right-end use few order statistics.
This is different  as compared to the Hill plot.
The first observation (not surprisingly) is that the empirical estimator $\hat T_n^{(1)}(1;1)$ is very sensitive with respect to the number of order statistics $k$, and is completely useless when plotted against large values of order statistics. The estimators motivated by the quasi-spectral representation are more "stable", even if the parameter $\alpha$ has to be estimated.

Figures 2
and
3 show Monte Carlo estimates of {\rm TDC} using $\hat T_n^{(1)}(1;1)$, $\hat T_n^{(2)}(1;1)$ (Figure 2)
and
$\hat T_n^{(2),\hat\alpha}(1;1)$ (Figure 3),
where the estimators are computed based on $k=5\%, 10\%, 20\%, 30\%, 40\%$ upper order statistics.
The parameter $\alpha$ in $\hat T_n^{(2),\hat\alpha}(1;1)$ is estimated using the Hill estimator based on $k_{\alpha}=5\%,10\%, 20\%,40\%$ of upper order statistics.

\begin{figure}
\label{fig:1}
\begin{center}
\includegraphics[width=0.9\textwidth,height=0.4\textheight]{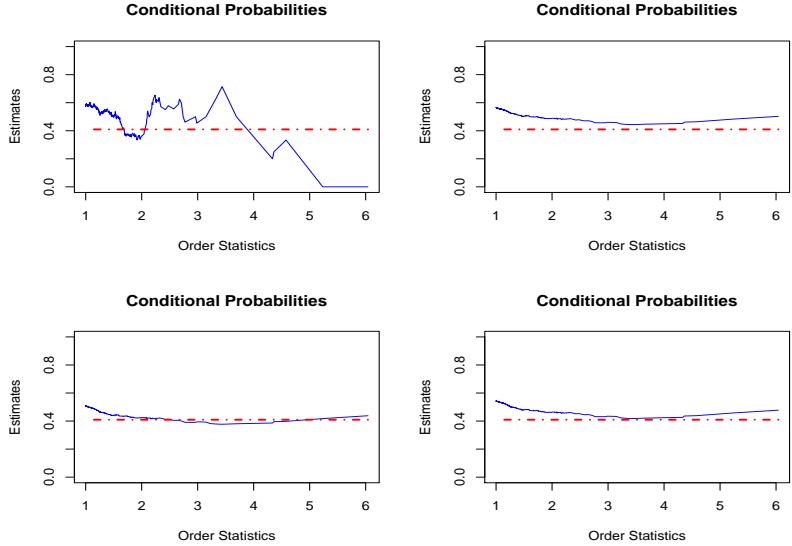}
\end{center}
\caption{\small{Estimation of {\rm TDC} for the model $Y=\phi X+\sigma |Z|$ with $\phi=0.8$, $\alpha=4$, $\sigma=0.1$. The dotted line shows the true value $\phi^{\alpha}$. Top line, left: estimator $\hat T_n^{(1)}(1;1)$; top line, right: estimator $\hat T_n^{(2)}(1;1)$; bottom line: estimators $\hat T_n^{(2),\hat\alpha}(1;1)$, where $\alpha$ is estimated using the Hill estimator based on 10\% (left picture) and 20\% (right picture) of order statistics.} }
\end{figure}

\begin{figure}
\label{fig:2a}
\begin{center}
\includegraphics[width=0.9\textwidth,height=0.35\textheight]{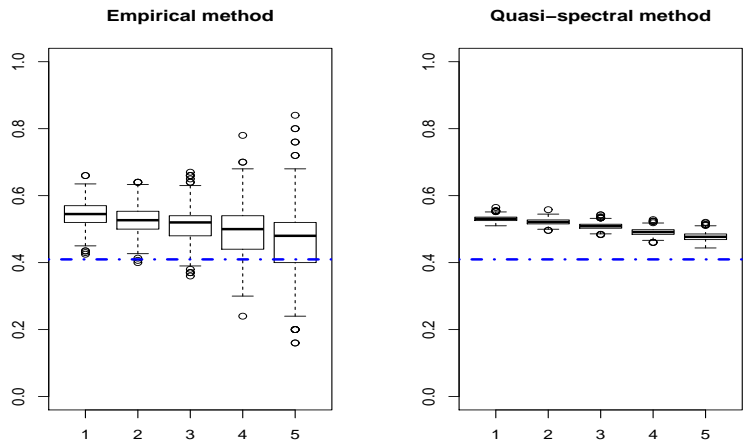}
\end{center}
\caption{\small{Estimation of {\rm TDC} for the model $Y=\phi X+\sigma |Z|$ with $\phi=0.8$, $\alpha=4$, $\sigma=0.1$. The dotted line shows the true value $\phi^{\alpha}$. Left panel: estimator $\hat T_n^{(1)}(1;1)$; right panel: estimator $\hat T_n^{(2)}(1;1)$. Each figure shows the boxplots for estimated values of the conditional probability computed for five different values of $k$. The first boxplot is computed based on $40\%$ of observations, the second one based on $30\%$ of observations, and the remaining ones based
on $20\%$, $10\%$ and $5\%$.  }}
\end{figure}

\begin{figure}
\label{fig:2b}
\begin{center}
\includegraphics[width=0.9\textwidth,height=0.45\textheight]{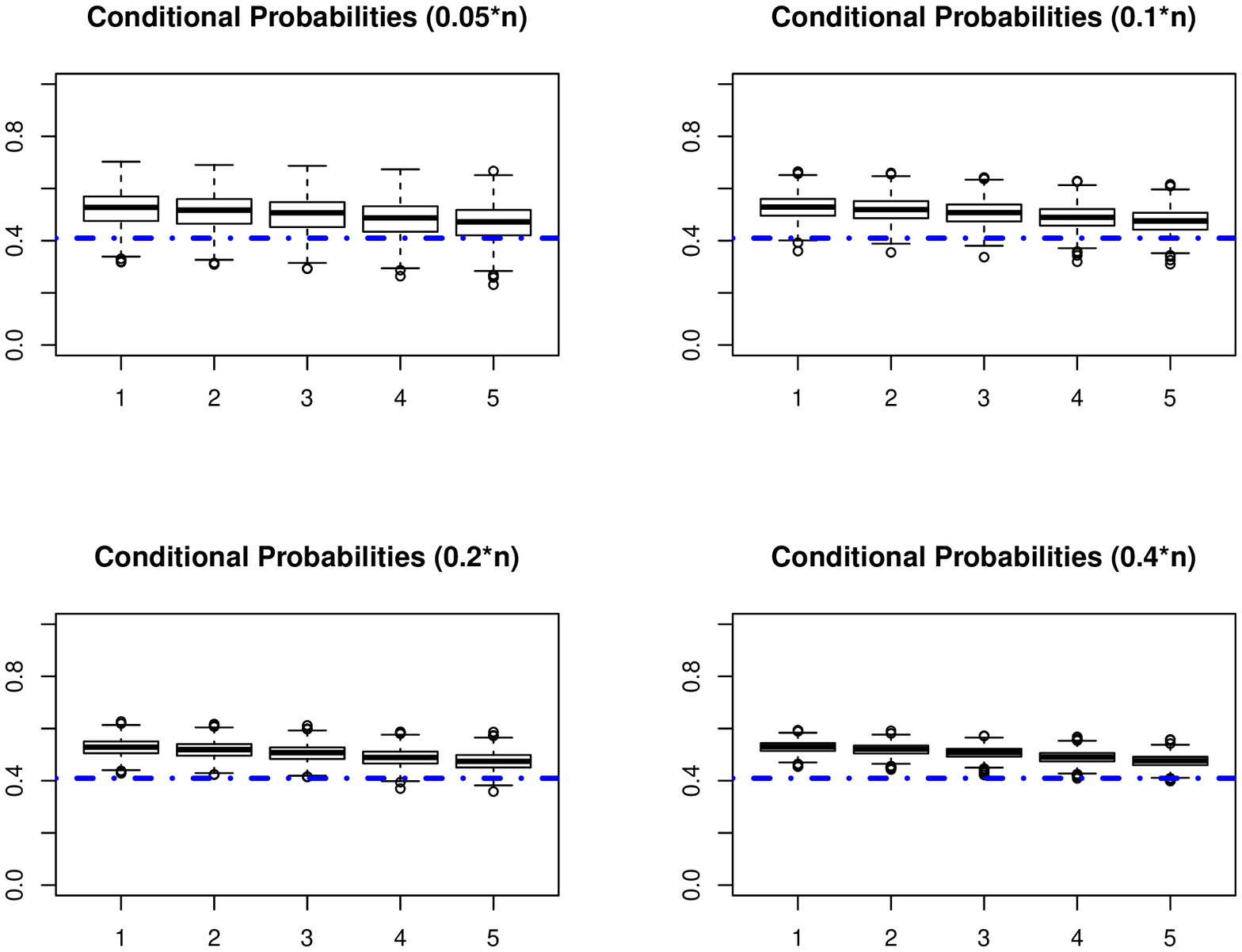}
\end{center}
\caption{\small{Estimation of {\rm TDC} for the model $Y=\phi X+\sigma |Z|$ with $\phi=0.8$, $\alpha=4$, $\sigma=0.1$. The dotted line shows the true value $\phi^{\alpha}$. Estimators $\hat T_n^{(2),\hat\alpha}(1;1)$ computed for $\hat\alpha$ obtained by the Hill estimator based on $5\%$ (top left), $10\%$ (top right), $20\%$ (bottom left) and $40\%$ (bottom right) order statistics. Each figure shows the boxplots for estimated values of the conditional probability computed for five different values of $k$. The first boxplot is computed based on $40\%$ of observations, the second one based on $30\%$ of observations, and the remaining ones based
on $20\%$, $10\%$ and $5\%$.  }}
\end{figure}

\subsection{Bivariate $t$}
We simulate 1000 observations from the bivariate $t$-distribution, that is $(X,Y)=\sqrt{W}(|Z_1|,|Z_2|)$, where $\alpha/W$ is chi-square with $\alpha=4$ degrees of freedom and $(Z_1,Z_2)$ are standard normal with correlation $\phi=0.9$.
In this case the tail dependence coefficient is $0.63$, see \cite{demarta:mcneiil:2004} .

\begin{figure}
\label{fig:1-t}
\begin{center}
\includegraphics[width=0.9\textwidth,height=0.4\textheight]{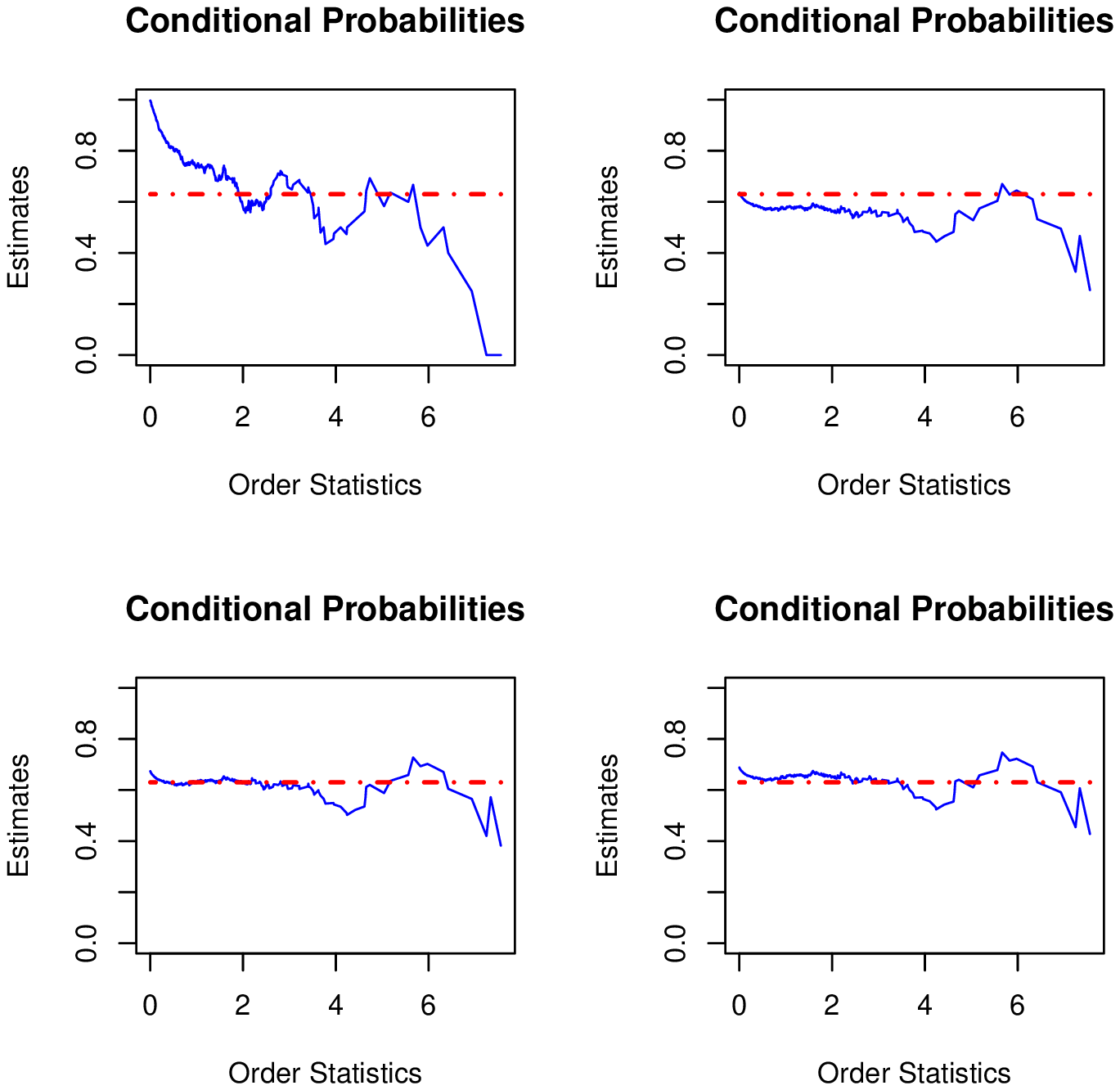}
\end{center}
\caption{\small{Estimation of {\rm TDC} for the bivariate $t$. Top line, left: estimator $\hat T_n^{(1)}(1;1)$; top line, right: estimator $\hat T_n^{(2)}(1;1)$; bottom line: estimators $\hat T_n^{(2),\hat\alpha}(1;1)$, where $\alpha$ is estimated using the Hill estimator based on 10\% (left picture) and 20\% (right picture) of order statistics.} }
\end{figure}

\begin{figure}
\label{fig:2a-t}
\begin{center}
\includegraphics[width=0.9\textwidth,height=0.35\textheight]{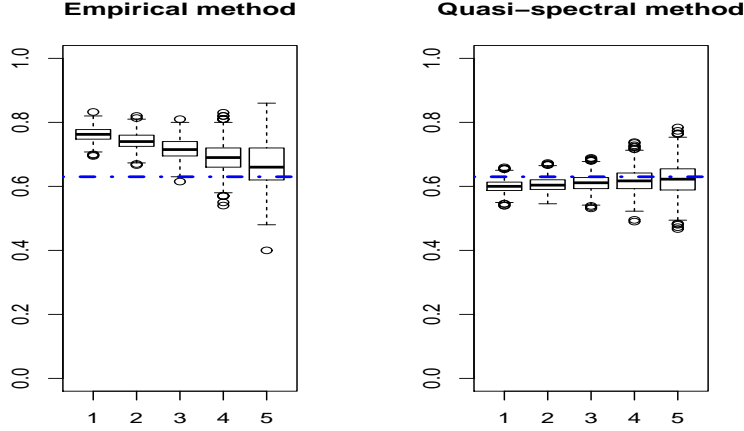}
\end{center}
\caption{\small{Estimation of {\rm TDC} for the bivariate $t$. Left panel: estimator $\hat T_n^{(1)}(1;1)$; right panel: estimator $\hat T_n^{(2)}(1;1)$. Each figure shows the boxplots for estimated values of the conditional probability computed for five different values of $k$. The first boxplot is computed based on $40\%$ of observations, the second one based on $30\%$ of observations, and the remaining ones based
on $20\%$, $10\%$ and $5\%$.  }}
\end{figure}

\begin{figure}
\label{fig:2b-t}
\begin{center}
\includegraphics[width=0.9\textwidth,height=0.45\textheight]{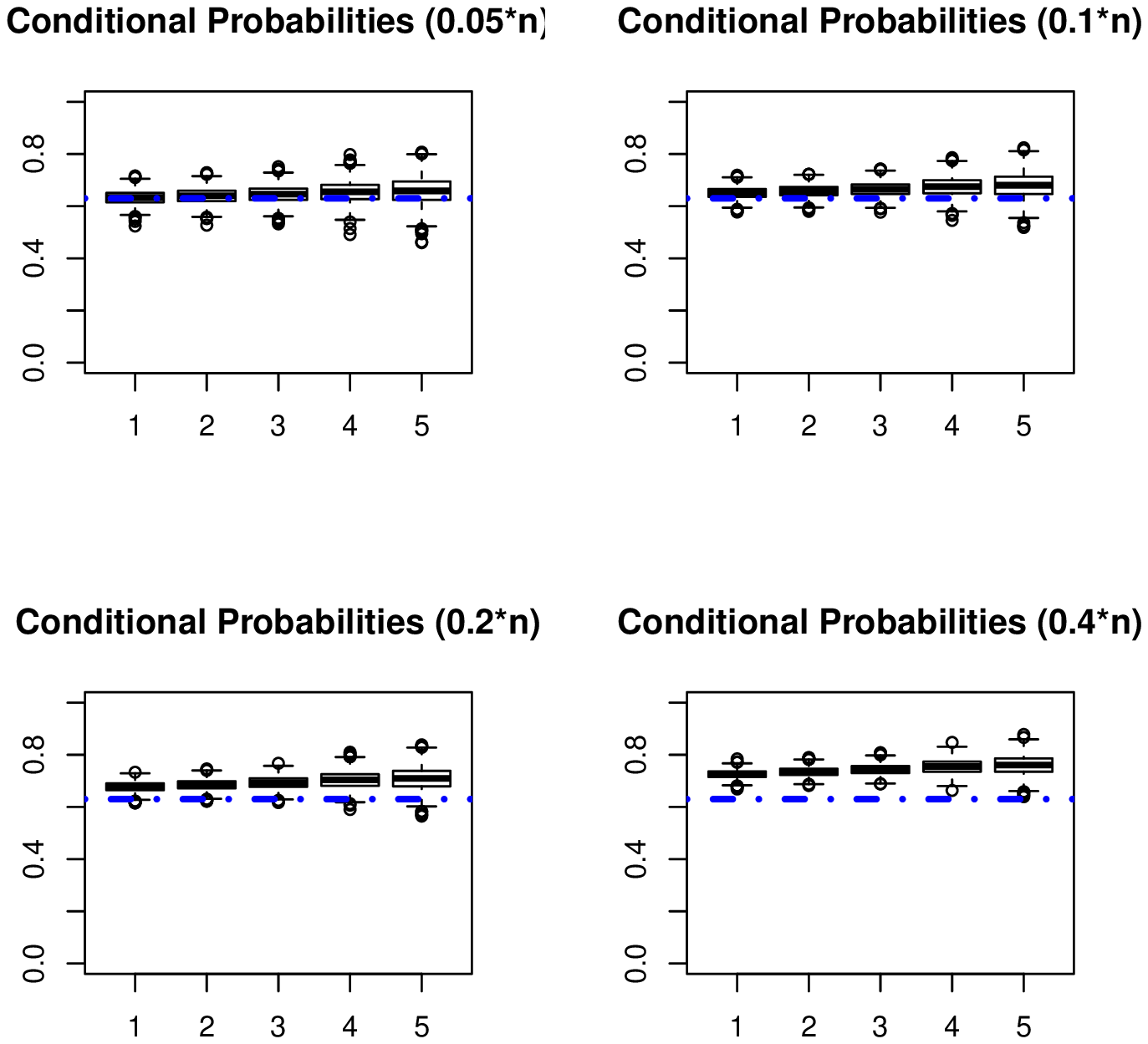}
\end{center}
\caption{\small{Estimation of {\rm TDC} for the bivariate $t$. Estimators $\hat T_n^{(2),\hat\alpha}(1;1)$ computed for $\hat\alpha$ obtained by the Hill estimator based on $5\%$ (top left), $10\%$ (top right), $20\%$ (bottom left) and $40\%$ (bottom right) order statistics. Each figure shows the boxplots for estimated values of the conditional probability computed for five different values of $k$. The first boxplot is computed based on $40\%$ of observations, the second one based on $30\%$ of observations, and the remaining ones based
on $20\%$, $10\%$ and $5\%$.  }}
\end{figure}

\section{Data Analysis}
We analyse absolut log-returns of S\&P500 and NASDAQ composite indices from January 2, 2013 until June 24, 2014. The scatter plot indicates strong dependence in the upper tail. This is confirmed by the estimation of the tail dependence coefficient. Again, the quasi-spectral method is less variable than the empirical one and robust with respect to the number $k$ of the order statistics and estimation of $\alpha$.

\begin{figure}
\label{fig:data-1}
\begin{center}
\includegraphics[width=0.9\textwidth,height=0.3\textheight]{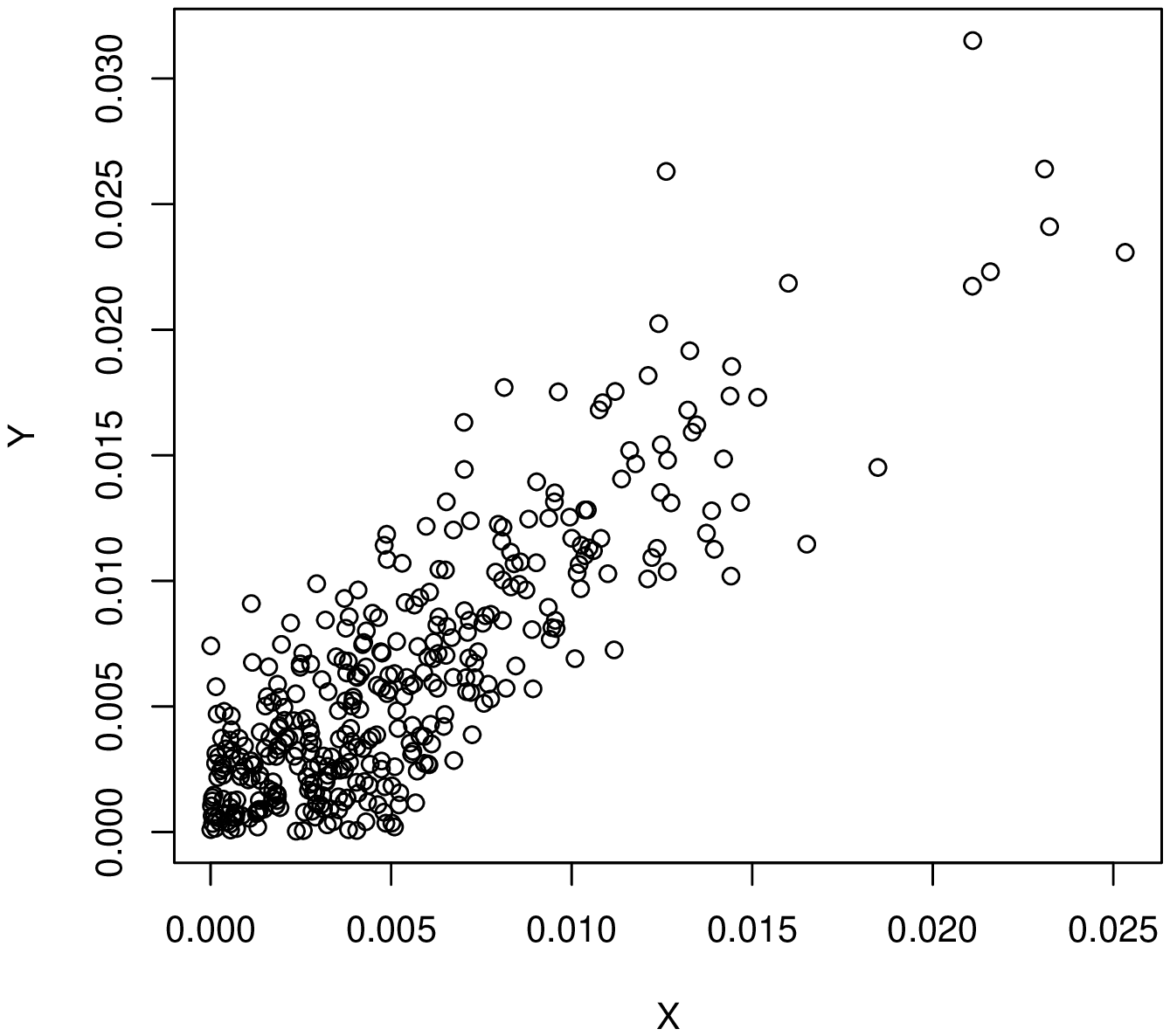}
\end{center}
\caption{\small{Scatter plot for S\&P vs. NASDAQ }}
\end{figure}

\begin{figure}
\label{fig:data-2}
\begin{center}
\includegraphics[width=0.9\textwidth,height=0.3\textheight]{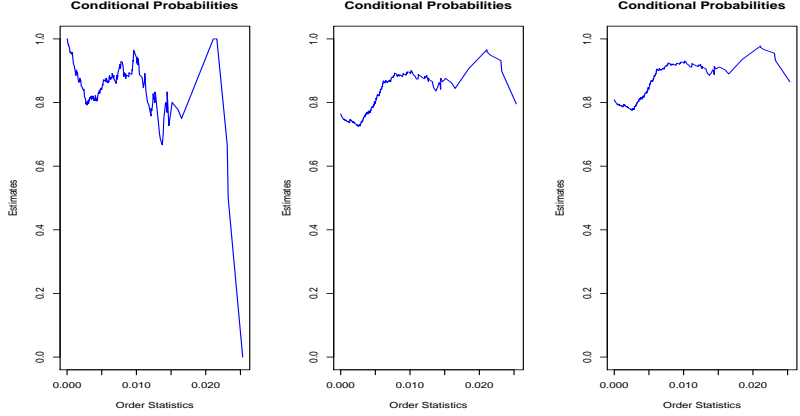}
\end{center}
\caption{\small{Estimation of TDC for S\&P and NASDAQ. Left plot: empirical method; middle plot: quasi-spectral method with $k_{\alpha}=0.1n$; right plot: quasi-spectral method with $k_{\alpha}=0.2n$ }}
\end{figure}

\section{Technical Details}\label{sec:technical-details}
We state the following lemma without a proof.
\begin{lemma}\label{lem:cond-exp}
Let $\vectorbold{X}$ be a regularly varying random vector such that all components are regularly varying with the same index $-\alpha$. Let $\psi:\Rset^d\to \Rset_+$ be homogenous with index $\gamma$ and assume that for some $\delta>0$,
\begin{align}\label{eq:condition-on-psi}
\int_C \psi^{1+\delta}(\vectorbold{v})\numult{} (\rmd \vectorbold{v})<\infty \; .
\end{align}
Then for $s>\epsilon$ and a relatively compact set $C$ in $\uncompact{d}$ we have
\begin{align*}
\lim_{x\to\infty}\frac{1}{\tail{F}(x)}\esp\left[\psi\left(\frac{\vectorbold{X}}{x}\right)\1{\{\vectorbold{X}\in sx C\}}\right]=s^{\gamma-\alpha}\int_C \psi(\vectorbold{v})\numult{}(\rmd \vectorbold{v})\; .
\end{align*}
\end{lemma}

\subsection{Proof of Proposition \ref{prop:quasi-spectral}}\label{sec:proof-of-quasispectral}
\begin{proof}
Since $\vectorbold{X}$ is regularly varying we have for $A\subseteq \Rset^{d-1}$,
\begin{align*}
\lim_{x\to\infty}\frac{\pr(x^{-1}\vectorbold{X}\in (y,\infty]\times A)}{\pr(X_{1}>x)}=
\frac{\numult{}((y,\infty]\times A)}{\numult{}((1,\infty]\times \Rset^{d-1})}\; .
\end{align*}
If moreover $y\geq 1$, the left hand side becomes the conditional probability
\begin{align*}
\lim_{x\to\infty}\pr(x^{-1}\vectorbold{X}\in (y,\infty]\times A\mid X_{1}>x)\; .
\end{align*}
In other words, conditionally on $X_{1}>x$, $x^{-1}\vectorbold{X}$ converges weakly to a random vector, say $\vectorbold{V}=(V_1,\ldots,V_d)$. Therefore, for any $f:\Rset^d\to \Rset$ bounded and continuous we have
\begin{align*}
\lim_{x\to\infty}\esp\left[f\left(x^{-1}\vectorbold{X}\right)\mid X_{1}>x\right]=\esp[f(\vectorbold{V})]\; .
\end{align*}
Now, let $g:\Rset^d\to \Rset$ be bounded and continuous. Then
\begin{align*}
\esp\left[g\left(\frac{X_1}{x},\frac{X_2}{X_1},\ldots,\frac{X_d}{X_1}\right)\mid X_1>x\right]=
\esp\left[f\left(\frac{X_1}{x},\frac{X_2}{x},\ldots,\frac{X_d}{x}\right)\mid X_1>x\right]\;,
\end{align*}
where $f(u_1,\ldots,u_d)=g(u_1,u_2/u_1,\ldots,u_d/u_1)$ is also bounded and continuous whenever $u_1\geq 1$. Hence,
 \begin{align*}
\lim_{x\to\infty}\esp\left[g\left(\frac{X_1}{x},\frac{X_2}{X_1},\ldots,\frac{X_d}{X_1}\right)\mid X_1>x\right]=
\esp\left[g(V_1,V_2/V_1,\ldots,V_d/V_1)\right]\; .
\end{align*}
Hence, conditionally on $X_1>x$,
\begin{align*}
\left(\frac{X_1}{x},\frac{X_2}{X_1},\ldots,\frac{X_d}{X_1}\right)
\end{align*}
converges in distribution to $(V_1,V_2/V_1,\ldots,V_d/V_1)=(V_1,\Theta_2,\ldots,\Theta_d)$.
It is obvious that $V_1$ has a standard Pareto distribution. We claim that
$V_1$ is independent of $(\Theta_2,\ldots,\Theta_d)$. Indeed, for $A_i\subseteq \Rset$, $i=2,\ldots, d$,
\begin{align*}
&\pr\left(\frac{X_1}{x}>y, \frac{X_2}{X_1}\in A_2, \cdots, \frac{X_d}{X_1}\in A_d\mid X_1>x\right)\\
&=\pr\left(\frac{X_2}{X_1}\in A_2, \cdots, \frac{X_d}{X_1}\in A_d\mid X_1>xy\right)\frac{\pr(X_1>xy)}{\pr(X_1>x)}\\
&\rightarrow \pr\left(\frac{V_2}{V_1}\in A_2, \cdots, \frac{V_d}{V_1}\in A_d\right)\pr(V_1>y)\;, \qquad x\to\infty\; .
\end{align*}
On the other hand,
\begin{align*}
&\lim_{x\to\infty}\pr\left(\frac{X_1}{x}>y, \frac{X_2}{X_1}\in A_2, \cdots, \frac{X_d}{X_1}\in A_d\mid X_1>x\right)= \pr\left(V_1>y, \frac{V_2}{V_1}\in A_2, \cdots, \frac{V_d}{V_1}\in A_d\right)\; .
\end{align*}
Hence, $(\Theta_2=V_2/V_1, \ldots, \Theta_d=V_d/V_1)$ and $V_1$ are independent.
\end{proof}

\subsection{Proof of Theorem \ref{theo:weakconv}}
The proof is relatively standard, but we provide it for completeness.
We start with the central limit theorem. Multivariate convergence follows by the Cramer-Wald device.
We prove the result only for $\process[G]_n(\cdot;\psi,C)$.
\begin{lemma}
Under the conditions of Theorem \ref{theo:weakconv}, for each $s\geq s_0$,
$\process[G]_n(s;\psi,C)$ converges in distribution to a centered normal random variable.
\end{lemma}
\begin{proof}
We prove the central limit theorem by checking Lindeberg's conditions.
Let
\begin{align*}
Z_{n,j}(s;C)=\frac{1}{\sqrt{n\tail{F}(u_n)}}\left\{\psi\left(\frac{X_{j}}{u_n},\frac{Y_{j}}{u_n}\right) \1{\{(X_{j},Y_{j})\in su_n C\}}-\esp\left[\psi\left(\frac{X_{j}}{u_n},\frac{Y_{j}}{u_n}\right) \1{\{(X_{j},Y_{j})\in su_n C\}}\right]\right\}\;
\end{align*}
so that $\process[G]_n(s;C)=\sum_{j=1}^n Z_{n,j}(s;C)$. Clearly, $\esp[Z_{n,j}(s;C)]=0$. Furthermore,
\begin{align*}
\var(\process[G]_n(s;\psi,C))&= \frac{1}{\tail{F}(u_n)}\esp\left[\psi^2\left(\frac{X}{ u_n},\frac{Y}{u_n}\right) \1{\{(X,Y)\in su_n C\}}\right] \\
&\phantom{=}-
\tail{F}(u_n)\left(\frac{1}{\tail{F}(u_n)}\esp\left[\psi\left(\frac{X}{u_n},\frac{Y}{u_n}\right) \1{\{(X,Y)\in su_n C\}}\right]\right)^2\; .
\end{align*}
Since $\tail{F}(u_n)\to 0$ as $n\to\infty$, Lemma \ref{lem:cond-exp} implies that the first term dominates and
$\lim_{n\to\infty}\var(\process[G]_n(s;\psi,C))$ exists. \\

Furthermore, noting that for arbitrary $\delta>0$ and any random variable $\1{\{|Y|>c\}}\leq |Y|^{\delta}/c^{\delta}$, we have
\begin{align*}
\esp[Z_{n,j}^2(s;C)\1{\{|Z_{n,j}|>\delta\}}]&\leq \frac{1}{ (n\tail{F}(u_n))^{\delta/2}}
\esp[|Z_{n,j}(s;C)|^{2+\delta}] \\
&\leq
\frac{K}{ (n\tail{F}(u_n))^{1+\delta/2}}
\esp\left[\psi^{2+\delta}\left(\frac{X_{1}}{u_n},\frac{X_{2}}{u_n}\right) \1{\{(X_{1},X_{2})\in su_n C\}}\right]
\end{align*}
and hence
\begin{align*}
\sum_{j=1}^n \esp[Z_{n,j}^2(s;C)\1{\{|Z_{n,j}|>\delta\}}]& \leq K (n\tail{F}(u_n))^{-\delta/2}
\left\{\frac{1}{\tail{F}(u_n)}\esp\left[\psi^{2+\delta}\left(\frac{X_{1}}{u_n},\frac{X_{2}}{u_n}\right) \1{\{(X_{1},X_{2})\in su_n C\}}\right]\right\}\; .
\end{align*}
Using Lemma \ref{lem:cond-exp} and since $\delta>0$, the expression on the right hand side converges to 0.
\end{proof}
\begin{lemma}
Under the conditions of Theorem \ref{theo:weakconv} the sequence of processes $\{\process[G]_n(\cdot;\psi,C)\}$, $n\geq 1$,  is tight in $\spaceD([s_0,\infty))$ equipped with the Skorokhod topology.
\end{lemma}
\begin{proof}In what follow, since the set $C$ is fixed, in our notation we omit a dependence on it, unless it is necessary.
For $s_0<s<t$, define $(s,t]u_nC=(su_nC)\setminus (t u_nC)$ and
\begin{align*}
U_{n,j}(s)=\psi\left(\frac{X_{j}}{u_n},\frac{Y_{j}}{u_n}\right)\1{\{(X_{j},Y_{j})\in su_nC\}}\;, \qquad
U_{n,j}^*(s)=U_{n,j}(s)-\esp[U_{n,j}(s)]\;,
\end{align*}
\begin{align*}
U_{n,j}(s,t)=U_{n,j}(s)-U_{n,j}(t)\;, \qquad
U_{n,j}^*(s,t)=U_{n,j}^*(s)-U_{n,j}^*(t) \; ,
\end{align*}
\begin{align*}
g_n(s;m)=\frac{1}{\tail{F}(u_n)}\esp\left[ \left|U_{n,j}(s)\right|^m\right]\;, \qquad g_n(s,t;m)=g_n(s;m)-g_n(t;m)\; .
\end{align*}
We note that $\lim_{n\to\infty}g_n(s;m)=s^{m\gamma-\alpha}\psi(C;m)$ uniformly on $[s_0,\infty)$. Then
\begin{align*}
\process[G]_n(s)-\process[G]_n(t)=\frac{1}{\sqrt{n\tail{F}(u_n)}}\sum_{j=1}^nU_{n,j}^*(s,t)\; ,
\end{align*}
where we write shortly $\process[G]_n(s)$ for $\process[G]_n(s;\psi,C)$.
We use Theorem 13.5 in \cite{billingsley:1999}. For $s_0<s_1<t<s_2$ we have
\begin{align}\label{eq:tightness-3}
&\esp\left[|\process[G]_n(s_1)-\process[G]_n(t)|^2|\process[G]_n(t)-\process[G]_n(s_2)|^2\right] \nonumber\\
&=\frac{1}{(n\tail{F}(u_n))^2}\sum_{j=1}^n\esp[\left(U_{n,j}^*(s_1,t)U_{n,j}^*(t,s_2)\right)^2]+\frac{1}{(n\tail{F}(u_n))^2}\sum_{\substack{i,j\\i\ne j}}^n\esp[\left(U_{n,i}^*(s_1,t)\right)^2]\esp[\left(U_{n,j}^*(t,s_2)\right)^2]\; .
\end{align}
By noting that for $s_1<t<s_2$ we have $U_{n,j}(s_1,t)U_{n,j}(t,s_2)=0$, we evaluate
\begin{align*}
&\left(U_{n,j}^*(s_1,t) U_{n,j}^*(t,s_2)\right)^2\\
&=U_{n,j}(s_1,t)  \esp^2[U_{n,j}(t,s_2)] + U_{n,j}(t,s_2)  \esp^2[U_{n,j}(s_1,t)] \\
&\phantom{=}-2 U_{n,j}(s_1,t)\esp[U_{n,j}(s_1,t)]\esp^2[U_{n,j}(t,s_2)]
-2 U_{n,j}(t,s_s)\esp[U_{n,j}(t,s_s)]\esp^2[U_{n,j}(s_1,t)]\\
&\phantom{=} + \esp^2[U_{n,j}(s_1,t)]\esp^2[U_{n,j}(t,s_2)] \;,
\end{align*}
so that
\begin{align*}
&\frac{1}{\tail{F}^2(u_n)}\esp\left[ \left( U_{n,j}^*(s_1,t)  U_{n,j}^*(t,s_2)\right)^2\right]
\leq
3 \tail{F}(u_n)  g_n^3(s_1,s_2;1)\; .
\end{align*}
Next, we deal with the second term in (\ref{eq:tightness-3}). For $s<t$ we have
\begin{align*}
&\esp[(U_{n,j}^*(s,t))^2]
\le 4\esp[(U_{n,j}(s)-U_{n,j}(t))^2]\; .
\end{align*}
Hence, the term is bounded by
\begin{align*}
&\frac{1}{(n\tail{F}(u_n))^2}\esp[U_{n,1}^*(s_1,t)]\esp[U_{n,1}^*(t,s_2)]
=K g_n^2(s_1,s_2;2)\; .
\end{align*}
The tightness follows.
\end{proof}
\subsection{Proof of Corollary \ref{cor:order-stats}}
The argument is similar to that of \cite{rootzen:2009}.
\begin{itemize}
\item By Theorem~\ref{theo:weakconv} and the Skorokhod representation theorem, there exists a
  probability space, a sequence of processes $\{\tilde{\process[G]}_n^*(\cdot),\tilde{\process[G]}_n(\cdot;\psi,C)\}$ and processes $\tilde{\process[G]}^*(\cdot)$, $\tilde{\process[G]}(\cdot;\psi,C)$ with the same
  distributions as, respectively, $\{\process[G]_n^*(\cdot),\process[G]_n(\cdot;\psi,C)\}$, $\process[G]^*(\cdot)$ and $\process[G](\cdot;\psi,C)$, such that
    \begin{align}
      \label{eq:skorokhod_representation}
      \tilde{\process[G]}_n^*(\cdot)\to \tilde{\process[G]}^*(\cdot)\;, \qquad \tilde{\process[G]}_n(\cdot;\psi,C)\to \tilde{\process[G]}(\cdot;\psi,C)
    \end{align}
    almost surely, uniformly on compact subsets of $[s_0,\infty)$. In what follows, for simplicity of
    notation we will write $\process[G]_n(\cdot)$, $\process[G]_n(\cdot;\psi,C)$, $\process[G](\cdot)$ and $\process[G](\cdot;\psi,C)$.
  \item Let $T_n^{\leftarrow}$ and $({\tilde{T}_n})^{\leftarrow}$ be the right continuous
    inverses of $T_n$ and $\tilde{T}_n$, respectively.
    Then, $T_n^{\leftarrow}(1)=1$, $({\tilde{T}_n})^{\leftarrow}(1)=X_{n:n-k}/\tepseq$ and, since
    $F$ is continuous, for all $s\in[\tail{F}(0)/\tail{F}(\tepseq),0]$,
    $T_n(T_n^{\leftarrow}(s))=s$.
  \item The (random) functions $\process[G]_n^*$ and $\tilde T_n^{\leftarrow}$ belong to
    $\spaceD$. Furthermore, their almost sure limits $\process[G]^*$ and $T^{\leftarrow}$ are
    continuous and $T^{\leftarrow}$ is strictly decreasing. Hence, the convergence
    (\ref{eq:skorokhod_representation}) and Theorem 3.1 in \cite{whitt:1980} imply that
    \begin{align*}
      \process[G]_n^*(T_n^{\leftarrow}(s))=\sqrt{k}\left\{\tilde{T}_n\circ
        T_n^{\leftarrow}(s)-s\right\}\to \process[G]^*(T^{\leftarrow}(s))
    \end{align*}
    almost surely, uniformly on compact subsets of $[s_0,\infty)$.
  \item Vervaat Lemma (\cite[Lemma A.0.2]{dehaan:ferreira:2006}) implies that
    \begin{align*}
      \sqrt{k}\left\{(\tilde{T}_n\circ T_n^{\leftarrow})^{\leftarrow}(s)-s\right\} \to
      - \process[G]^*(T^{\leftarrow}(s))
    \end{align*}
    almost surely, uniformly on compact subsets of $[s_0,\infty)$.
  \item Assumption~(\ref{eq:tep-derivee-uniform-convergence}) implies that $T_n$ is
    continuous and strictly decreasing in a neighborhood of~1. Thus, there exists $\epsilon>0$ such
    that $T_n\circ (\tilde{T}_n)^{\leftarrow}(s) = (\tilde{T}_n\circ
    T_n^{\leftarrow})^{\leftarrow}(s)$ for $s\in(1-\epsilon,1+\epsilon)$ and
    \begin{align}
      \label{eq:randomlevels-vervaat}
      \sqrt{k}\left\{T_n\circ (\tilde{T}_n)^{\leftarrow}(s)-s\right\}\to
      -\process[G]^*(T^{\leftarrow}(s))\; ,
    \end{align}
    almost surely uniformly with respect to $s \in (1-\epsilon,1+\epsilon)$.
  \item Since $\statinterseq\to\infty$ and
    $(\tilde{T}_n)^{\leftarrow}(1)=\orderstat{n:n-\statinterseq}/\tepseq$,
    (\ref{eq:randomlevels-vervaat}) implies that
    $T_n(\orderstat{n:n-\statinterseq}/\tepseq)$ converges almost surely to~1. Since
    $T(1)=1$ and $T_n$ converges uniformly to $T$ in a neighborhood of~1,
    this implies that $\orderstat{n:n-\statinterseq}/\tepseq$ converges almost surely to~1.
  \item By Taylor's expansion, there exists $\varsigma_n$ such that $|\varsigma_n-1|\leq
    |(\tilde{T}_n)^{\leftarrow}(1)-1|$ and
    \begin{align}
      T_n((\tilde{T}_n)^{\leftarrow}(1))-1 & = T_n((\tilde{T}_n)^{\leftarrow}(1)) -
      T_n(T_n^{\leftarrow}(1)) \nonumber      \\
      & = T_n' (\varsigma_n) \left\{(\tilde{T}_n)^{\leftarrow}(1) -
        T_n^{\leftarrow}(1)\right\} \nonumber      \\
      & = T_n' (\varsigma_n) \left\{ \orderstat{n:n-\statinterseq}/\tepseq-1\right\} \; .
      \label{eq:randomlevels-expansion}
    \end{align}
  \item Thus, (\ref{eq:tep-derivee-uniform-convergence}), (\ref{eq:randomlevels-vervaat}) and
    (\ref{eq:randomlevels-expansion}) yield that
    \begin{align}
      \label{eq:randomlevels-conclusion}
      \sqrt{k} \left\{\frac{X_{n:n-k}}{\tepseq}-1\right\} \to \frac1\alpha\process[G]^*(1)\; ,
    \end{align}
    almost surely.
  \item Since the convergences $\process[G]_n(\cdot;\psi,C)\to\process[G](\cdot;\psi,C)$ and (\ref{eq:randomlevels-conclusion}) hold almost
    surely, they hold jointly. Coming back to the original probability space, we obtain the joint weak
    convergence.
  \end{itemize}

\subsection{Proof of Theorem \ref{theo:fclt-random-ext}}

\begin{proof}
  Denote $\hat T_n(s;\psi,C) = \tilde T_n(sX_{n:n-k}/\tepseq;\psi,C)$, where $\tilde T_n$ and $\hat T_n$ are the tail empirical functions defined in (\ref{eq:ted}) and (\ref{eq:ted-random}), respectively. Then, by the homogeneity property (\ref{eq:homogeneity-scaling}),
\begin{align*}
\hat{\process[G]}_n(s;\psi,C) & = \process[G]_n(sX_{n:n-k}/\tepseq;\psi,C) \\
    & \phantom{ = } + s^{\gamma-\alpha}\sqrt k  \left\{T(X_{n:n-k}/\tepseq;\psi,C) -
      T(1;\psi,C) \right\}=I_1(s)+s^{\gamma-\alpha}I_2(s)\; .
\end{align*}
By Corollary \ref{cor:order-stats} \begin{align}\label{eq:clt-orderstat}
\sqrt{k} \left\{\frac{X_{n:n-k}}{\tepseq}-1\right\} \convdistr \frac{1}{\alpha}\process[G]^*(1)\;,
\end{align}
jointly with $\process[G]_n(\cdot;\psi,C)$.
In particular, $X_{n:n-k}/u_n$ converges in probability to 1.
Thus, by Theorem~\ref{theo:weakconv}, the term $I_1$ converges weakly to $\process[G](\cdot;\psi,C)$, while
by the delta method the term $I_2(s)$ converges weakly to
\begin{align*}
\frac{1}{\alpha}T'(1;\psi,C)\process[G]^*(1)\; .
\end{align*}
This finishes the proof of (\ref{eq:fclt-random-ext}). Furthermore,
  \begin{align*}
    \hat{\hat{\process[G]}}_n(s;\psi,C) & = \left(\frac{X_{n:n-k}}{\tepseq}\right)^{-\gamma}
    \process[G]_n(sX_{n:n-k}/\tepseq;\psi,C) \\
    & \phantom{ = } + \sqrt k
    \left(\frac{X_{n:n-k}}{\tepseq}\right)^{-\gamma}
    \{T_n(sX_{n:n-k}/\tepseq;\psi,C) -
    T(sX_{n:n-k}/\tepseq;\psi,C) \}    \\
    & \phantom{ = } + \sqrt k
    \left\{\left(\frac{X_{n:n-k}}{\tepseq}\right)^{-\gamma} -
      1\right\} s^{\gamma-\alpha}T(X_{n:n-k}/\tepseq;\psi,C) \\
    & \phantom{ = } + \sqrt k s^{\gamma-\alpha} \left\{T(X_{n:n-k}/\tepseq;\psi,C) -
      T(1;\psi,C) \right\}\\
      &=\left(\frac{X_{n:n-k}}{\tepseq}\right)^{-\gamma}I_1(s)+J_1(s)+s^{\gamma-\alpha}J_2(s)+s^{\gamma-\alpha}I_2(s) \; .
  \end{align*}
Again, by Theorem \ref{theo:weakconv} and $X_{n:n-k}/u_n\convprob 1$, the first term converges weakly to $\process[G](\cdot;\psi,C)$.
The second term vanishes by (\ref{eq:bias}).  Furthermore, the delta method,
the first order Taylor expansion of $T(\cdot;\psi,C)$ around 1 and (\ref{eq:clt-orderstat}) yield that $s^{\gamma-\alpha}(J_2(s)+I_2(s))$ converges to
\begin{align*}
-\frac{\gamma}{\alpha}s^{\gamma-\alpha}T(1;\psi,C)\process[G]^*(1)+\frac{1}{\alpha}s^{\gamma-\alpha}T'(1;\psi,C)\process[G]^*(1)\;.
\end{align*}
The convergence (\ref{eq:fclt-random-random-ext}) is proven.
\end{proof}

\section{Additional comments and future research}
We finish our paper by addressing several technical issues and discussing directions of future research.

\begin{enumerate}
\item We assume regular variation of a vector $(X,Y)$ since we work under general framework of estimating (\ref{eq:to-estimate}). In specific examples, like conditional tail expectation, it is enough to assume that the limit $\lim_{x\to\infty}x^{-1}\esp[Y\mid X>x]$ exists and is strictly positive. This is done precisely in \cite{cai:einmahl:haan:zhou:2012}.
\item In expense of additional technical considerations one can study tightness with respect to a class of sets $C\in {\mathcal C}$, which in particular will imply tightness with respect to $y$ in case of the conditional tail distribution.
\item The results are meaningful in case of extremal dependence, that is when the exponent measure is not concentrated on axes. In case of extremal independence, if one wants to estimate quantities like the conditional tail distribution or conditional tail expectation, a different scaling is required. We will address this issue in a following paper, based upon the ideas developed in \cite{heffernan:resnick:2007}, \cite{fougeres:soulier:2012}, \cite{kulik:soulier:2013a,kulik:soulier:2014}.
\item The quasi-spectral method should be compared with semiparametric or parametric ones. It could be particularly attractive in case of time series where very few parametric models for multivariate extremes are available.
\item We would like to address estimation of conditional tail expectation in a context of multivariate time series, using the tools developed in \cite{drees:rootzen:2010}.
\item It is a common practice in extreme value theory to standardize marginals. Assume that we have a positive bivariate vector $(X,Y)$ with marginal distribution functions $F_X$ and $F_Y$. Define
    \begin{align*}
    Q_X(t)=\left(\frac{1}{\tail{F}_X}\right)^{\leftarrow}(t)\;, \qquad Q_Y(t)=\left(\frac{1}{\tail{F}_Y}\right)^{\leftarrow}(t)\; .
    \end{align*}
    Then $Z=Q_X^{\leftarrow}(X)$ and $W=Q_Y^{\leftarrow}(Y)$ are standard Pareto. All results in the paper remain valid if one assumes that $(Z,W)$ is regularly varying (with index $\alpha=1$). If $(V_1',V_1'\Theta_2')$ is the quasi-spectral decomposition of $(Z,W)$, then $V_1'$ is standard Pareto, however $\Theta_2'$ still contains information about the marginal behaviour. For example, if we start with $(X,Y)$ being regularly varying with $-\alpha$ and $(V_1,V_1\Theta_2)$ is its quasi-spectral decomposition, then $\Theta_2'=\Theta_2^{\alpha}$. In other words, by transforming marginals we do not avoid the problem of estimating $\alpha$ in (\ref{eq:tef-2-estim}).
\end{enumerate}




\section*{Acknowledgement}
Research supported by NSERC grant.
\bibliographystyle{plain}
\bibliography{bib}


\end{document}